\documentclass[aps,prl,reprint,superscriptaddress,longbibliography,nofootinbib,secnumarabic]{revtex4-2}
\usepackage{amsmath,amssymb,mathtools,graphicx,booktabs,microtype}
\usepackage[T1]{fontenc}
\usepackage{lmodern}
\usepackage{bbm}
\usepackage{xcolor}
\usepackage[normalem]{ulem}
\usepackage[colorlinks=true,
  linkcolor={blue!60!green},
  citecolor={green!50!yellow!75!black},
  urlcolor={blue!80!black}]{hyperref}
\newcommand{\cL}{\mathcal L}
\newcommand{\cLgad}[2]{\cL_{\scriptscriptstyle #1,#2}}
\newcommand{\comment}[1]{}
\newcommand{\cN}{\mathcal N}
\newcommand{\cM}{\mathcal M}
\newcommand{\cS}{\mathcal S}

\newcommand{\cK}{\mathcal K}
\newcommand{\cJ}{\mathcal J}
\newcommand{\cQ}{\mathcal Q}
\newcommand{\cR}{\mathcal R}
\newcommand{\cB}{\mathcal B}

\newcommand{\eps}{\varepsilon}
\newcommand{\1}{\mathbbm{1}}
\newcommand{\id}{\operatorname{id}}
\newcommand{\Tr}{\operatorname{Tr}}
\newcommand{\supp}{\operatorname{supp}}
\newcommand{\Span}{\operatorname{span}}
\newcommand{\cA}{\mathcal A}
\newcommand{\DA}{D^{\mathcal A}}
\newcommand{\Dreg}{D^{\mathrm{reg}}}
\newcommand{\ket}[1]{|#1\rangle}
\newcommand{\bra}[1]{\langle#1|}
\newcommand{\kb}[1]{|#1\rangle\!\langle#1|}
\newcommand{\op}[1]{\left\lVert#1\right\rVert_\infty}
\newcommand{\normone}[1]{\left\lVert#1\right\rVert_1}
\newcommand{\cpleq}{\preceq}
\newcommand{\dd}{\mathrm d}
\newcommand{\smin}{s_{\min}}
\newcommand{\bmin}{\beta^{(\smin)}_{\eps,t}}
\newcommand{\bet}{\beta_{\eps,t}}
\newcommand{\perr}{p_{\mathrm{err},t}}

\newenvironment{proof}[1][Proof]
 {\par\noindent\textbf{#1.}\ }
 {\unskip\nobreak\hfill\(\blacksquare\)\par}
\newtheorem{theorem}{Theorem}
\newtheorem{proposition}{Proposition}
\newtheorem{lemma}{Lemma}
\newtheorem{corollary}{Corollary}
\newcommand{\nn}{\nonumber}

\usepackage{array}[=2016-10-06]
\usepackage{enumitem}
\hypersetup{pdftitle={Discriminating Lindbladian Dynamics},pdfauthor={Robert Salzmann, Ludovico Lami, Martin B. Plenio and Susana F. Huelga}}
\begin{document}

\makeatletter
\let\letter@section\section
\let\letter@subsection\subsection
\let\letter@raggedcolumn@sw\raggedcolumn@sw
\@booleantrue\raggedcolumn@sw
\let\letter@normalsize\normalsize
\renewcommand{\normalsize}{%
  \letter@normalsize
  \setlength{\abovedisplayskip}{6pt plus 1pt minus 1pt}%
  \setlength{\belowdisplayskip}{6pt plus 1pt minus 1pt}%
  \setlength{\abovedisplayshortskip}{6pt plus 1pt minus 1pt}%
  \setlength{\belowdisplayshortskip}{6pt plus 1pt minus 1pt}%
}
\normalsize
\edef\letter@footinskip{\the\skip\footins}
\skip\footins=12pt plus 2pt minus 2pt
\edef\letter@interfootnotelinepenalty{\the\interfootnotelinepenalty}
\interfootnotelinepenalty=10000
\renewcommand{\section}{%
  \@startsection{section}{1}{\z@}%
    {2.2ex \@plus .5ex \@minus .2ex}%
    {1.5ex}%
    {\normalfont\small\bfseries\centering}%
}
\renewcommand{\subsection}{%
  \@startsection{subsection}{2}{\z@}%
    {2.2ex \@plus .5ex \@minus .2ex}%
    {1.5ex}%
    {\normalfont\small\bfseries\centering}%
}
\makeatother
\title{Discriminating Lindbladian Dynamics}
\author{Robert Salzmann}
\email{robert.salzmann@rwth-aachen.de}
\affiliation{Institute for Quantum Information, RWTH Aachen University, Germany}
\affiliation{Inria, ENS de Lyon, Universit\'e Claude Bernard Lyon 1, LIP, 69342 Lyon cedex 07, France}
\author{Ludovico Lami}
\email{ludovico.lami@gmail.com}
\affiliation{Scuola Normale Superiore, Piazza dei Cavalieri 7, 56126 Pisa, Italy}
\author{Martin B. Plenio}
\email{martin.plenio@uni-ulm.de}
\affiliation{Institut f\"ur Theoretische Physik \& IQST, Albert-Einstein-Allee 11, Universit\"at Ulm, D-89081 Ulm, Germany}
\author{Susana F. Huelga}
\email{susana.huelga@uni-ulm.de}
\affiliation{Institut f\"ur Theoretische Physik \& IQST, Albert-Einstein-Allee 11, Universit\"at Ulm, D-89081 Ulm, Germany}

\begin{abstract}
The fundamental task of quantum hypothesis testing, which consists in discriminating quantum objects, has been extensively studied for quantum states and channels. Here, we investigate the fundamental limits of discriminating an arbitrary pair of quantum Markovian dynamics generated by Lindbladians, focusing on the roles of total exposure time and temporal resolution. We allow adaptive choices of evolution durations and arbitrary instantaneous controls between them.
We characterise the optimal type-II error exponent (Stein exponent) as the largest instantaneous increase of quantum relative entropy, and give an algebraic criterion for its finiteness in terms of the Lindbladians.
Remarkably, an infinite exponent is equivalent to a vanishing type-II error probability at sufficiently large finite total time, for every positive type-I tolerance.
While some dynamics admit a finite protocol with zero type-II error, others require increasingly fine temporal control.
For the latter, the error vanishes as the minimum allowed evolution duration, or resolution time, tends to zero at fixed sufficiently large total time, but remains bounded away from zero at every fixed positive resolution time.
We achieve this using quantum Zeno protocols, establishing temporal resolution as a resource for discrimination.
When the null hypothesis is the identity semigroup, an infinite asymptotic error exponent has a stronger consequence: at a fixed finite total exposure time, the type-II error can be made arbitrarily small while keeping the type-I error exactly zero. This is impossible in the corresponding classical setting.
Discriminating the identity semigroup from generalised amplitude damping on a qubit provides examples of all three regimes: a finite asymptotic error exponent, zero type-II error attainable at positive resolution time, and vanishing type-II error requiring the resolution time to tend to zero.
\end{abstract}
\maketitle

\section{Introduction}

Quantum hypothesis testing extends the statistical problem of deciding between competing hypotheses to quantum systems, and can be seen as a foundation for a broad range of quantum information protocols~\cite{Hayashi2017,Watrous2018,Wilde2017,KhatriLamiWilde2025}.
In the binary setting, the experimenter must decide between two possible descriptions of a quantum system, designated as the null hypothesis and alternative hypothesis.
In state discrimination, the experimenter receives independent copies of one of two states and performs a joint binary measurement to determine which state was supplied.
In the asymmetric setting, the quantum Stein lemma identifies the quantum relative entropy of the two states as the optimal error exponent per copy in the limit of many copies~\cite{Hiai1991,Ogawa2000,WangWilde2019States}.
On the other hand, when the hypotheses are two quantum processes, described by two fixed quantum channels, the experimenter also chooses the inputs, may retain an auxiliary quantum system, and can adapt later operations to earlier outcomes~\cite{Chiribella2009}.
Extending the state-discrimination result to quantum channels, the optimal asymmetric error exponent per channel use is characterised by the regularised channel relative entropy~\cite{Wilde2020,WangWilde2019Channels,Fang2020,Bergh2024}.

For quantum Markovian dynamics, the evolution channels form a semigroup parametrised by the evolution duration and generated by a Lindbladian~\cite{Lindblad1976,Gorini1976}.
Compared with fixed-channel discrimination, the experimenter has the additional freedom to choose the evolution durations adaptively, applying a control operation or performing a measurement after each period before allowing the system to evolve further.

Related approaches to identifying quantum dynamics include hypothesis testing from continuous measurement records, using likelihood-ratio methods~\cite{Tsang2012}, measurements of emitted fields~\cite{Molmer2015,KiilerichMolmer2018}, and adaptive stopping rules~\cite{Gasbarri2024}, as well as testing whether Lindbladian dissipation lies below or above prescribed thresholds~\cite{BaoCaiShang2026}.
Related work studies Hamiltonian~\cite{Wiebe2014,Holzaepfel2015,Huang2023,Bakshi2024,Hu2025} and Lindbladian learning~\cite{Franca2024,Franca2025TimeDependent,Ivashkov2026,Mobus2026,Arad2026,LewisTangWright2026,ChenYu2026,Romanov2026}, i.e.~estimating an unknown generator from observations of its dynamics.

The possibility of improving discrimination by intervening during the system's time evolution connects this problem to quantum control~\cite{WisemanMilburn2009,Koch2016,AltafiniTicozzi2012}.
For example, quantum Zeno protocols use repeated interventions to control the evolution and have been applied to interaction-free measurement and channel discrimination~\cite{Kwiat1995,ZhouYung2017,HasenoehrlWolf2022}.
For Lindbladian pairs with identical jump operators and different Hamiltonian parts, the conditions for perfect discrimination within finite total exposure time using quantum error correction have been characterised~\cite{ChenYuan2019}.

Here, we study binary discrimination for arbitrary pairs of Lindblad generators, allowing adaptive choices of evolution durations and instantaneous control operations.
We ask how reliably the hypotheses can be distinguished within a given total exposure time.
We characterise the optimal asymmetric error exponent and give an algebraic characterisation for its finiteness.
An infinite exponent is equivalent to the possibility of making the type-II error arbitrarily small at a fixed, sufficiently large finite exposure time, for every positive type-I tolerance.
We quantify the gain from finer temporal resolution through achievable error bounds and a converse valid for all adaptive protocols.
These bounds hold uniformly for sufficiently large exposure times and sufficiently short resolution times.
When the null hypothesis is the identity semigroup and the exponent is infinite, we determine the minimum total exposure time at which the type-II error can be made arbitrarily small by allowing arbitrarily fine temporal resolution, including at zero type-I tolerance.

\section{The Lindbladian discrimination problem}
We now make this discrimination task precise.
Consider a finite-dimensional quantum system $A$ evolving under one of two \emph{Lindblad generators}
\begin{equation}
 \cL_i(X)=G_iX+XG_i^*+\sum_jL_j^{(i)}XL_j^{(i)*},
 \label{eq:generator}
\end{equation}
where, for $i=1,2$, the \emph{jump operators} $L_j^{(i)}$ and \emph{drift operator} $G_i$ satisfy $G_i+G_i^*+\sum_jL_j^{(i)*}L_j^{(i)}=0$.
The corresponding dynamics is described by the channels $e^{s\cL_i}$, $s\ge0$, which together form the quantum dynamical semigroup generated by $\cL_i$.
The null hypothesis is that the unknown generator is $\cL_1$, and the alternative hypothesis is that it is $\cL_2$.\footnote{The formulation in \eqref{eq:generator} is equivalent to
$\cL_i(X)=-i[H_i,X]+\sum_j(L_j^{(i)}XL_j^{(i)*}-\frac12\{L_j^{(i)*}L_j^{(i)},X\})$
for some self-adjoint operator $H_i$, where we choose $G_i=-iH_i-\frac12\sum_jL_j^{(i)*}L_j^{(i)}$.}
We model the unknown dynamics as memoryless and time-homogeneous under all allowed interventions.

We assume that the experimenter can let the system evolve under the unknown dynamics for chosen durations, subject to a total evolution time budget $t\ge0$.
Each protocol has a finite upper bound on the number of evolution intervals, although no common bound is imposed across the protocols considered in the optimisation.
The most general protocol $\mathsf P$ available to the experimenter starts with an input state $\omega_{AE}$ on $AE$, where $E$ is an arbitrary finite-dimensional auxiliary system which is unaffected by the unknown dynamics and accessible to the experimenter's control operations.
Under hypothesis $i$, each evolution interval of duration $s$ applies $e^{s\cL_i}\otimes\id_E$ to the current state of $AE$.
Between these evolution intervals, the experimenter prepares the next input by applying a quantum channel $\cA_k$ to $AE$, identical under both hypotheses (Fig.~\ref{fig:protocol})~\cite{Chiribella2009}.
The number of calls $n$ and their durations $s_1,\ldots,s_n$ may depend on preceding measurement outcomes and are optimised together with the input state, control operations and final measurement, subject to $\sum_{k=1}^n s_k\le t$ for every outcome sequence.
Furthermore, to model practical limits on how rapidly the experimenter can intervene, we require each evolution interval to last at least a \emph{resolution time} $\smin\ge0$.

\begin{figure*}[t]
 \centering
 \includegraphics[width=0.94\textwidth]{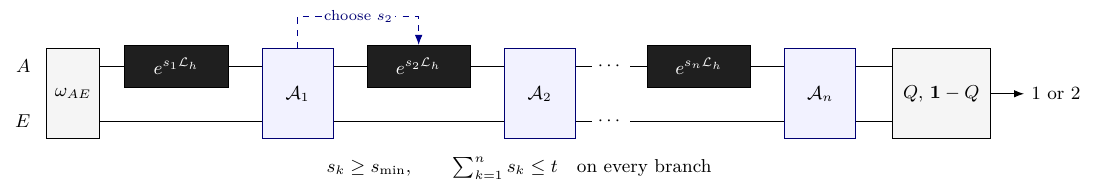}
 \caption{An adaptive discrimination protocol for an unknown Lindblad generator $\cL_h$, with fixed $h\in\{1,2\}$.
 Starting from $\omega_{AE}$, evolutions on $A$ alternate with control channels $\cA_k$ on $AE$, followed by the binary POVM $\{Q,\1-Q\}$.
 The outcome associated with $Q$ gives a decision in favour of the null hypothesis $h=1$.
 The number of calls $n$ and their durations $s_k$ may depend on preceding outcomes, with $\sum_k s_k\le t$ and $s_k\ge\smin$ for every outcome sequence.}
 \label{fig:protocol}
\end{figure*}

For such a protocol $\mathsf P$, let $\rho_{AE}^{(i),\mathsf P}$ denote the final accessible state on $AE$ under hypothesis $i$.
The protocol ends with a binary POVM $\{Q,\1-Q\}$ on $AE$, where $0\le Q\le\1$ and the outcome associated with $Q$ represents a decision in favour of the null hypothesis.
The corresponding type-I and type-II error probabilities are
\begin{align*}
 \alpha(\mathsf P)=\Tr\left((\1-Q)\rho_{AE}^{(1),\mathsf P}\right),\qquad
 \beta(\mathsf P)=\Tr\left(Q\rho_{AE}^{(2),\mathsf P}\right),
\end{align*}
respectively.
For an allowed type-I error probability $0\le\eps<1$, we define the optimal type-II error probability at total time $t$ and resolution time $\smin$ by
\begin{equation}
 \bmin(\cL_1\|\cL_2)\coloneqq
 \inf_{\mathsf P}\left\{\beta(\mathsf P):
 \begin{array}{l}
 \alpha(\mathsf P)\le\eps,\\
 \sum_k s_k\le t,\ s_k\ge\smin
 \end{array}\right\}.
 \label{eq:beta}
\end{equation}
The infimum ranges over such finite protocols, so a zero optimal value need not be attained by any individual protocol.
Furthermore, we write \begin{equation}\bet(\cL_1\|\cL_2)\coloneqq\beta^{(0)}_{\eps,t}(\cL_1\|\cL_2)\end{equation}
 for the \emph{unrestricted optimal type-II error probability}.\footnote{We omit the qualifier \emph{unrestricted} when the context is clear.}

\section{The Lindbladian Stein exponent}

Analogously to the usual setting of quantum state and channel discrimination, the optimal type-II error probability $\bet(\cL_1\|\cL_2)$ decays at least exponentially in the total exposure time $t$ for distinct generators and any fixed positive tolerance $0<\eps<1$.\footnote{One way to see this is to apply the quantum Stein lemma to repeated experiments of a suitable fixed duration~\cite{Hiai1991,Ogawa2000}.}
We therefore define the \emph{Lindbladian Stein exponent} by
\begin{equation}
 \gamma(\cL_1\|\cL_2)\coloneqq
 \lim_{\eps\to0}\liminf_{t\to\infty}
 -\frac1t\log\bet(\cL_1\|\cL_2).
 \label{eq:gamma}
\end{equation}
Due to Hiai and Petz's result~\cite{Hiai1991}, the above quantity can be thought of as a notion of relative entropy between Lindbladians. Our first result gives a formula to compute it. To present such a formula,
we first recall that the quantum relative entropy~\cite{Umegaki1962} of two states $\rho,\sigma$ is defined as $D(\rho\|\sigma)=\Tr[\rho(\log\rho-\log\sigma)]$ when $\supp\rho\subseteq\supp\sigma$, and $+\infty$ otherwise.\footnote{All logarithms are to base two.}

\begin{theorem}[Differential characterisation]\label{thm:letter-rate}
For any two finite-dimensional Lindbladians,
\begin{align}
 \gamma(&\cL_1\|\cL_2)=\sup_{\rho_{AE},\sigma_{AE}}
 \left.\frac{\dd}{\dd s}
 D\!\left(e^{s\cL_1}(\rho_{AE})\middle\|e^{s\cL_2}(\sigma_{AE})\right)\right|_{s=0}\nn\\
 &=\sup_{\rho_{AE},\sigma_{AE}}\Bigl(
 \Tr\bigl[\cL_1(\rho_{AE})
\bigl(\log\rho_{AE}-\log\sigma_{AE}\bigr)\bigr]\nn\\
 &\qquad-\Tr\bigl[\rho_{AE}\,\Omega_{\sigma_{AE}}
 \bigl(\cL_2(\sigma_{AE})\bigr)\bigr]\Bigr),
 \label{eq:differential}
\end{align}
where the supremum ranges over all finite-dimensional auxiliary systems $E$ and positive definite states $\rho_{AE},\sigma_{AE}$.\footnote{For simplicity, we write $\cL_i$ and $e^{s\cL_i}$ for $\cL_i\otimes\id_E$ and $e^{s\cL_i}\otimes\id_E$ in \eqref{eq:differential}.}
Here $\Omega_{\sigma_{AE}}$ is the derivative of the operator logarithm at $\sigma_{AE}$, defined in \eqref{sm:omega}.
\end{theorem}

Hence, the Lindbladian Stein exponent is the largest instantaneous rate at which the two dynamics increase the relative entropy between inputs that may already differ because of earlier adaptive steps.
Only the additional distinguishability generated by the next evolution interval is counted.

For the achievability part of the proof of Theorem~\ref{thm:letter-rate}, we use parallel channel-discrimination strategies with a fixed evolution duration, implemented sequentially using the auxiliary system.
The quantum Stein lemma for channel discrimination~\cite{WangWilde2019Channels} and the relative-entropy chain rule~\cite{Fang2020} show that these strategies achieve the amortised channel relative entropy divided by the evolution duration.
Optimising over this duration gives a lower bound on the Lindbladian Stein exponent.
For the converse, the telescoping-sum argument of~\cite{Wilde2020} gives the same quantity as an upper bound for all adaptive protocols.
Finally, the semigroup property and integration of the entropy derivative show that this common bound equals the largest instantaneous rate of increase in relative entropy, yielding the differential formula.
The full proof is given in Appendix~\ref{sec:rate}.

\section{An algebraic criterion for a finite Stein exponent}\label{sec:letter-finite-exponent}
Next we show that finiteness of the Lindbladian Stein exponent, $\gamma(\cL_1\|\cL_2)$, is fully characterised by explicit algebraic conditions on the jump operators and drifts of the two Lindbladians.
Furthermore, this is equivalent to positivity of the optimal type-II error probability $\beta_{\eps,t}(\cL_1\|\cL_2)$ for all $t\ge0$ at some fixed tolerance $0<\eps<1$.
In order to state this result, we define the space $\cS_2\coloneqq\Span\{\1,L_j^{(2)}:j\}$.

\begin{theorem}[Stein exponent dichotomy]\label{thm:letter-span}
The following conditions are equivalent:
\begin{enumerate}
 \item $\gamma(\cL_1\|\cL_2)<\infty$,
 \item\label{item:letter-span-algebraic} $L_j^{(1)}\in\cS_2$ for all $j$, and $G_1-G_2\in\cS_2$,
 \item There exists $0<\eps<1$ such that $\beta_{\eps,t}(\cL_1\|\cL_2)>0$ for all $t\ge0$.
\end{enumerate}
Furthermore, if these conditions hold, then for $0\le\eps<1$ and all $t\ge0$,
\begin{equation}
 \bet(\cL_1\|\cL_2)\ge 2^{-[t\gamma(\cL_1\|\cL_2) +h(\eps)]/(1-\eps)},
 \label{eq:finitefloor}
\end{equation}
where $h(\eps)=-\eps\log\eps-(1-\eps)\log(1-\eps)$ is the binary entropy.
\end{theorem}

In particular, for distinct generators with finite $\gamma(\cL_1\|\cL_2)$, the optimal type-II error probability decays exponentially, so $-\log\bet(\cL_1\|\cL_2)=\Theta(t)$ as $t\to\infty$ at fixed $0<\eps<1$ (Appendix~\ref{sec:finite}).
Failure of the algebraic conditions has a stronger consequence than an arbitrarily large asymptotic exponent: it allows the optimal type-II error probability to vanish at finite total time.
The remaining question is when zero optimal type-II error can be achieved at positive resolution time and when it requires the experimenter to apply control operations at arbitrarily short intervals.

\begin{figure*}[t]
\centering
\begin{minipage}[t]{.435\textwidth}
 \centering
 \includegraphics[page=1,width=0.87\linewidth]{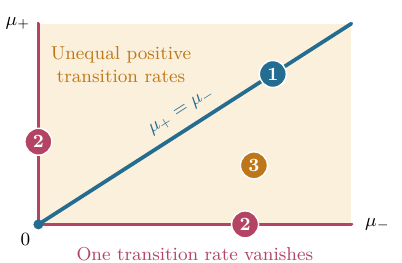}
 \par\smallskip
 {\small (a) Three regimes in the transition rate plane.\par}
\end{minipage}\hfill
\begin{minipage}[t]{.55\textwidth}
 \centering
 \includegraphics[page=2,width=.87\linewidth]{discrimination-regimes.pdf}
 \par\smallskip
 {\small (b) Error versus resolution time.\par}
\end{minipage}
\caption{Illustration of the three regimes in Table~\ref{tab:regimes}, taking as null hypothesis the identity semigroup, i.e. $\cL_1=0$, and as alternative hypothesis a member of the generalised amplitude damping family defined in Section~\ref{sec:letter-gad}, i.e. $\cL_2= \cLgad{\mu_-}{\mu_+}$.
(a) The transition rates $\mu_-$ and $\mu_+$ parametrise the plane. Equal transition rates, including the origin, give Regime~1 (blue), exactly one vanishing transition rate gives Regime~2 (red), and unequal positive transition rates give Regime~3 (orange).
(b) Schematic optimal type-II error probabilities for one fixed generator from each regime, at fixed tolerance $0<\eps<1$ and a common sufficiently large total time $t$.
The displayed range $0<\smin\le s_0$ is chosen sufficiently small for all three examples.
Values at $\smin=0$ denote the unrestricted optimal type-II error probabilities. Here $\gamma=\gamma(0\|\cLgad{\mu_-}{\mu_+})$ is the Stein exponent, and all errors are subject to $\alpha\le\eps$.}
\label{fig:regimes}
\end{figure*}

\section{Two regimes for zero error discrimination -- temporal resolution as a resource}\label{sec:letter-infinite-exponent}
To address this question, we first note that, by Theorem~\ref{thm:letter-span}, for every $0<\eps<1$, we have $\gamma(\cL_1\|\cL_2)=\infty$ if and only if there exists $t_\eps\ge0$ such that, for all $t\ge t_\eps$,
\begin{equation*}
 \beta_{\eps,t}(\cL_1\|\cL_2)=0.
\end{equation*}
Interestingly, there are two operationally distinct regimes in which $\beta_{\eps,t}(\cL_1\|\cL_2)=0$ occurs. In the first, for all positive resolution times $\smin>0$, there exists a finite protocol attaining zero type-II error while satisfying $\alpha\le\eps$, i.e.
\begin{align}
    \label{eq:FiniteResolutionPerfectBeta}
    \beta^{(\smin)}_{\eps,t}(\cL_1\|\cL_2) = 0,
\end{align}
for all sufficiently large $t$, where the required total time may depend on $\smin$.
Clearly, this then also implies zero type-II under general protocols, i.e. $\beta_{\eps,t}(\cL_1\|\cL_2)=0$. By Lemma~\ref{lem:support-attainment}, such finite protocols exist if and only if
\begin{align}
\label{eq:FailureSuportInclusion}
\supp J^{(1)}_{AA'}(s)\not\subseteq\supp J^{(2)}_{AA'}(s)
\end{align}
for some $s>0$.
Here $J^{(i)}_{AA'}(s)\coloneqq\left(e^{s\cL_i}\otimes\id_{A'}\right)\left(\kb{\Phi}\right)$ denotes the Choi state of $e^{s\cL_i}$, with $\Phi=d^{-1/2}\sum_{a=1}^d\ket a_A\ket a_{A'}$ being the maximally entangled state and $A'\simeq A$.
This support violation occurs, for example, when we take as null hypothesis the identity semigroup and as alternative hypothesis the amplitude damping semigroup, discussed as a special case of generalised amplitude damping in Section~\ref{sec:letter-gad}.
For a fixed evolution duration, \eqref{eq:FailureSuportInclusion} is the Choi-support criterion from the usual setting of binary channel discrimination.
It characterises the possibility of attaining zero type-II error at any prescribed positive type-I tolerance with finitely many channel uses, for both parallel and adaptive protocols~\cite{Wilde2020,KrawiecPawelaPuchala2021}.

Strikingly, in the Lindbladian discrimination setting considered here, we can have support inclusion for all $s>0$,
\begin{equation}
 \supp J^{(1)}_{AA'}(s)\subseteq\supp J^{(2)}_{AA'}(s),
 \label{eq:choi}
\end{equation}
while $\gamma(\cL_1\|\cL_2)=\infty$ and hence, by Theorem~\ref{thm:letter-span}, the optimal type-II error probability satisfies $\beta_{\eps,t}(\cL_1\|\cL_2)=0$ for all sufficiently large $t$.
In this second operational regime we have $\bmin(\cL_1\|\cL_2)>0$ for all positive resolution times $\smin>0$, yet the error tends to zero as $\smin\to0$ at fixed sufficiently large total time.
In this sense, arbitrarily short resolution times become a resource that enables arbitrarily reliable Lindbladian discrimination.
An example satisfying \eqref{eq:choi} with an infinite Lindbladian Stein exponent is obtained by taking as null hypothesis the identity semigroup and as alternative hypothesis the generalised amplitude damping semigroup with unequal positive transition rates (Section~\ref{sec:letter-gad}).
The following theorem quantifies this phenomenon by giving upper and lower bounds on the optimal type-II error probability as the resolution time tends to zero.
 
\begin{theorem}[Gain from shorter resolution time]\label{thm:letter-resolution}
Suppose \eqref{eq:choi} holds and $\gamma(\cL_1\|\cL_2)=\infty$.
Then for all $0<\eps<1$, there exist $s_0>0$, $C_\eps\ge c_\eps>0$ and $t_\eps<\infty$, depending only on the generators and $\eps$, such that
\begin{equation}
 \left(\frac{\smin}{2s_0}\right)^{C_\eps t}
 \le\bmin(\cL_1\|\cL_2)
 \le\left(\frac{\smin}{2s_0}\right)^{c_\eps t}
 \label{eq:resolutionlaw}
\end{equation}
for all $t\ge t_\eps$ and $0<\smin\le s_0$.
Moreover, the threshold $t_\eps$ can be chosen to satisfy $t_\eps\le K\,(\,1+\log(1/\eps))$, for some constant $K\ge 0$ depending only on $\cL_1$ and $\cL_2$.

In particular, this gives $\bmin(\cL_1\|\cL_2)>0$ for all $\smin>0$ and $t\ge0$ and further $\beta_{\eps,t}(\cL_1\|\cL_2) = \lim_{s_{\min}\to 0} \beta^{(s_{\min})}_{\eps,t}(\cL_1\|\cL_2)  =0$ for all $t\ge t_\eps.$
\end{theorem}

The achievability proof of Theorem~\ref{thm:letter-span} uses a Zeno construction to show that violation of the algebraic conditions in item~\ref{item:letter-span-algebraic} implies zero optimal type-II error at sufficiently large total exposure time.
The same construction also yields the upper bound in Theorem~\ref{thm:letter-resolution}.
If item~\ref{item:letter-span-algebraic} fails, we choose orthogonal initial and target states coupled by a jump operator under the null hypothesis or by the drift difference.
Under the alternative hypothesis, the jump operators cause no transition from the initial state to the target, and a common Hamiltonian control cancels the drift transition in the Zeno limit.
Repeated projections onto their span then preserve the initial state in this limit, conditioned on every projection succeeding, while the null hypothesis still permits a transition to the target.
At fixed total exposure time, one run accepts the null hypothesis with probability at least $p_0>0$ under the null hypothesis and $O(\smin)$ under the alternative hypothesis.
Independent repetitions reduce the type-I error to at most $\eps$ within total exposure time $t_\eps=O(1+\log(1/\eps))$ and yield the upper bound in \eqref{eq:resolutionlaw} (Appendix~\ref{sec:zeno}).

For comparison, \cite{ChenYuan2019} combines quantum error correction methods with quantum Zeno dynamics to achieve perfect discrimination between Lindbladians with identical jump operators when the difference of their Hamiltonian parts lies outside $\Span\{\1,L_j,L_j^*,L_j^*L_k:j,k\}$.
Our asymmetric criterion applies to arbitrary generator pairs.

Table~\ref{tab:regimes} summarises the three  asymmetric discrimination regimes established in Sections~\ref{sec:letter-finite-exponent} and \ref{sec:letter-infinite-exponent}: one with finite Lindbladian Stein exponent and two with infinite exponent.
Figure~\ref{fig:regimes} illustrates all three within the generalised amplitude damping family.

\begin{table*}[t]
\caption{General classification of asymmetric Lindbladian discrimination for a fixed tolerance $0<\eps<1$, with regime numbers matching Figure~\ref{fig:regimes}.
Zero refers to type-II error subject to $\alpha\le\eps$.
In Regime~2 the required total time may depend on $\smin$. In Regime~3 it can be chosen independently of sufficiently small $\smin$.}
\label{tab:regimes}
\centering
\renewcommand{\arraystretch}{1.10}
\setlength{\tabcolsep}{4.5pt}
\begin{tabular}{@{}>{\raggedright\arraybackslash}p{.17\textwidth}>{\raggedright\arraybackslash}p{.26\textwidth}>{\raggedright\arraybackslash}p{.32\textwidth}>{\raggedright\arraybackslash}p{.19\textwidth}@{}}
\toprule
Regime & Characterisation & Operational consequence & Examples\\
\midrule
\textcolor[HTML]{246D91}{\textbf{1}}\quad Finite Stein exponent & $\gamma(\cL_1\|\cL_2)<\infty$, equivalently the span conditions of Theorem~\ref{thm:letter-span} & $\bet(\cL_1\|\cL_2)>0$ for all $t\ge0$, even with arbitrarily short resolution time & Identity vs depolarisation\\
\addlinespace[2.5pt]
\textcolor[HTML]{B34464}{\textbf{2}}\quad Zero at positive resolution time & $\gamma(\cL_1\|\cL_2)=\infty$ and $\supp J^{(1)}_{AA'}(s)\not\subseteq\supp J^{(2)}_{AA'}(s)$ for some $s>0$ & For all $\smin>0$, a finite protocol attains $\beta=0$ for sufficiently large $t$ & Identity vs amplitude damping, depolarisation vs identity, \\
\addlinespace[2.5pt]
\textcolor[HTML]{BE7718}{\textbf{3}}\quad Zero only as the resolution time vanishes & $\gamma(\cL_1\|\cL_2)=\infty$ and $\supp J^{(1)}_{AA'}(s)\subseteq\supp J^{(2)}_{AA'}(s)$ for all $s>0$ & $\bmin(\cL_1\|\cL_2)>0$ for all $\smin>0$ and $t\ge0$, but $\bet(\cL_1\|\cL_2)=0$ for sufficiently large $t$ & Identity vs generalised amplitude damping with unequal transition rates\\
\bottomrule
\end{tabular}
\end{table*}

\comment{For the achievability part of the proof of Theorem~\ref{thm:letter-resolution}, we construct a protocol that isolates a two-dimensional subspace using an auxiliary system and repeated projections.
A common coherent control cancels the drift under the alternative hypothesis along a transition outside its jump span.
Under the null hypothesis, that transition still occurs.
A block of fixed total time then has an event with probability at least $p_0>0$ under the null hypothesis and at most $C\smin$ under the alternative hypothesis, for constants independent of $\smin$.
Repeating the block amplifies the separation and gives the upper bound in \eqref{eq:resolutionlaw}.
All failed projections enter the decision, so the probabilities are unconditional (Appendix~\ref{sec:zeno}).
For the converse, analyticity of the Choi states and their common identity limit bound the entropy acquired per unit total time by a logarithm of the inverse interval (Appendix~\ref{sec:log}).}

\section{An exact threshold for trivial null hypothesis}

We now consider the special case in which the null hypothesis is the identity semigroup, i.e. $\cL_1=0$.
In this case, the measure defined in \eqref{eq:g} below quantifies the violation of the algebraic condition for $\gamma(0\|\cL_2)<\infty$ in item~\ref{item:letter-span-algebraic} of Theorem~\ref{thm:letter-span}.
As Theorem~\ref{thm:letter-identity} below shows, this measure determines the exact threshold time $T_\eps$ at which the optimal type-II error probability first vanishes.
When $\gamma(0\|\cL_2)=\infty$, this threshold remains finite even at $\eps=0$, so the type-II error can be made arbitrarily small at a fixed total time while the type-I error is exactly zero.

To define this mentioned measure, we write the generator under the alternative hypothesis as $\cL\equiv\cL_2$, with drift $G$ and jump operators $L_j$ as in \eqref{eq:generator}, and set
\begin{equation}
 g\coloneqq\inf_{C\in\Span\{\1,L_j:j\}}\op{G-C},
 \label{eq:g}
\end{equation}
where $\|\cdot\|_\infty$ denotes the operator norm.
Note that by Theorem~\ref{thm:letter-span} we have
\begin{align*}
    \gamma(0\|\cL)&=\infty
 \quad\Longleftrightarrow\quad
 g>0,
\end{align*}
which can be seen by noting that for $\cL_1=0$ the jump condition in item~\ref{item:letter-span-algebraic} is trivially satisfied and the drift condition holds if and only if $g=0$.

\begin{theorem}[Identity null hypothesis]\label{thm:letter-identity}
If $g>0$, then for all $0\le\eps<1$,
\begin{equation}
 \beta_{\eps,t}(0\|\cL)=0
 \quad\Longleftrightarrow\quad
 t\ge T_\eps\coloneqq\frac{\arccos\sqrt\eps}{g}.
 \label{eq:threshold}
\end{equation}
\end{theorem}

For a purely Hamiltonian alternative hypothesis, \eqref{eq:g} gives $g=(\lambda_{\max}(H)-\lambda_{\min}(H))/2$.
At $\eps=0$, \eqref{eq:threshold} therefore recovers the known minimum total exposure time for perfectly discriminating Hamiltonian evolution from the identity~\cite{Childs2000}.

The achievability part of Theorem~\ref{thm:letter-identity} refines the construction used for the upper bound in Theorem~\ref{thm:letter-resolution}.
Under the identity null hypothesis, every projection succeeds and the control produces an exact rotation.
We choose the initial and target states so that the control transfers population between them with transition probability $\sin^2(gt)$.
Hence the type-I error probability is $\alpha=\cos^2(gt)$, giving $\alpha=\eps$ at total exposure time $T_\eps$, while $\beta=O(\smin)$.
A matching converse proves optimality (Appendix~\ref{sec:identity}).

Theorem~\ref{thm:letter-identity} also determines the threshold for symmetric discrimination, in which the tester tries to minimise the average probability over type I and type II error.
For equal priors, define \begin{align*}\perr(0,\cL)\coloneqq\inf_{\mathsf P}(\alpha(\mathsf P)+\beta(\mathsf P))/2,
\end{align*} where the protocols have total time at most $t$ and resolution time zero.
For $g>0$, the theorem gives
\begin{equation*}
 \perr(0,\cL)=0\quad\iff\quad t\ge T_0=\frac{\pi}{2g}.
\end{equation*}
Indeed, vanishing symmetric error implies $\beta_{\eps,t}(0\|\cL)=0$ for all $\eps>0$, whereas $\beta_{0,t}(0\|\cL)=0$ already implies vanishing symmetric error.

\comment{For the achievability part of the proof of Theorem~\ref{thm:letter-identity}, we use duality in \eqref{eq:g} to find orthogonal unit vectors $\psi,v$ on $AE$ with vanishing jump matrix elements and $\bra{v}(G\otimes\1_E)\ket{\psi}=g$. An auxiliary system of dimension $d$ suffices.
After each short call, we apply the common control generated by $H_c=-ig\ket{v}\!\bra{\psi}+ig\ket{\psi}\!\bra{v}$ and then project onto their span, keeping the limiting state under the alternative hypothesis on the line $\psi$.
Under the null hypothesis, given by the identity semigroup, the same control rotates the state with transition probability $\sin^2(gt)$.
At $T_\eps$, protocols with positive resolution time have exactly $\alpha=\eps$ and $\beta=O(\smin)$.
For the converse, a change of jump representation leaves a residual drift of norm $g$.
The no-jump trajectory under the alternative hypothesis retains positive weight and separates from the state under the null hypothesis by an angle at most $gt$, enforcing a positive error below $T_\eps$ (see \eqref{sm:thresholdfloor}).
When the alternative hypothesis is purely Hamiltonian, $g$ is half the spectral width. The same formula also applies to entirely dissipative dynamics.}

The possibility of a vanishing optimal symmetric error at finite total time is a genuinely quantum feature.
For any two finite-state classical Markov generators $L_1,L_2$, the optimal symmetric error satisfies
\begin{equation*}
 p^{\mathrm{class}}_{\mathrm{err},t}(L_1,L_2)\ge\tfrac12e^{-\Lambda t}>0,
\end{equation*}
where $\Lambda\ge0$ is a finite constant depending only on $L_1,L_2$.
To understand this bound, note that under either classical hypothesis there is a positive probability that the unknown process makes no transition during the experiment.
Even if the experimenter observes every transition, its absence is compatible with both hypotheses, so an error cannot be excluded.
The probability weight shared by these records gives the bound above (Proposition~\ref{prop:classical}).
Quantum evolution can instead change the relative amplitudes of a superposition even when no jump occurs, allowing coherent control to separate the two hypotheses.
This use of fast control connects our discrimination protocols to quantum metrology~\cite{Sekatski2017,DemkowiczDobrzanski2017,Zhou2018} and Lindblad estimation~\cite{Gardner2025}.

\section{Examples of the three regimes}\label{sec:letter-examples}

In the following, we illustrate the results of the preceding sections by considering specific families of quantum semigroups. We begin with the \emph{depolarising semigroup} on a $d$-dimensional quantum system, with $d\ge2$, defined by
\begin{align*}
 e^{t\cL_{\mathrm{dep}}}(\rho)
 :=e^{-t}\rho+(1-e^{-t})\1/d,
\end{align*}
for $\rho$ being a state, and which is generated by the Lindladian $\cL_{\mathrm{dep}}(X):=\Tr(X)\1/d-X$.
Taking as null hypothesis the identity semigroup, i.e. $\cL_1=0$, and as alternative hypothesis the depolarising semigroup, we explicitly calculate the Lindbladian Stein exponent in Appendix~\ref{sec:depolarising}:
\begin{equation*}
 \gamma(0\|\cL_{\mathrm{dep}})=(1-d^{-2})\log e<\infty.
\end{equation*}
By Theorem~\ref{thm:letter-span}, the optimal type-II error probability is therefore positive at every finite total time $t\ge0$ for $0\le\eps<1$, placing this example in Regime~1.

\subsection{Generalised amplitude damping channel}\label{sec:letter-gad}

Another parametrised family of interest is the \emph{generalised amplitude damping semigroup} on a qubit.
Its jump operators are
\begin{equation*}
 L_-=\sqrt{\mu_-}\ket0\!\bra1,\qquad
 L_+=\sqrt{\mu_+}\ket1\!\bra0,
\end{equation*}
with relaxation rate $\mu_-\ge0$ and excitation rate $\mu_+\ge0$, and its coherent part vanishes, $H=0$, which results 
in the Lindbladian
\begin{align*}
 \cLgad{\mu_-}{\mu_+}(X)={}&\mu_-\left(\ket0\!\bra1X\ket1\!\bra0
 -\frac12\{\kb1,X\}\right)\\
 &+\mu_+\left(\ket1\!\bra0X\ket0\!\bra1
 -\frac12\{\kb0,X\}\right).
\end{align*}
The two terms describe competing relaxation $\ket1\to\ket0$ and excitation $\ket0\to\ket1$, at transition rates $\mu_-$ and $\mu_+$, respectively.
When $\mu_+=0$ and $\mu_->0$, only relaxation remains and we recover the usual \emph{amplitude damping semigroup}.

We keep as null hypothesis the identity semigroup and take as alternative hypothesis a member of this family.
For fixed $0<\eps<1$, varying the transition rates realises all three regimes in Figure~\ref{fig:regimes} and Table~\ref{tab:regimes}:
Equal transition rates $\mu_+=\mu_-$, including the origin, give a finite Stein exponent, $\gamma(0\|\cLgad{\mu_-}{\mu_+})<\infty$, and a positive optimal type-II error probability at every finite total time (Regime~1, Appendix~\ref{sec:equal-transition-rates}).

If exactly one transition rate vanishes, the support violation \eqref{eq:FailureSuportInclusion} occurs for every $s>0$.
A repeated support test then attains zero type-II error, $\bmin(0\|\cLgad{\mu_-}{\mu_+})=0$, for every positive resolution time $s_{\min}>0$ once the total time $t\ge 0$ is sufficiently large.
In particular, $\gamma(0\|\cLgad{\mu_-}{\mu_+})=\infty$ (Regime~2, Appendix~\ref{sec:amplitude-damping}).

For unequal positive transition rates, the Lindbladian Stein exponent is infinite, $\gamma(0\|\cLgad{\mu_-}{\mu_+})=\infty$, while the support inclusion \eqref{eq:choi} holds (Regime~3, Appendix~\ref{sec:generalised-amplitude-damping}).
Thus the optimal type-II error probability remains positive at every fixed positive resolution time $s_{\min}>0$ and finite total time $t\ge 0$, but tends to zero as the resolution time vanishes at fixed sufficiently large total time (Theorem~\ref{thm:letter-resolution}).

For unequal transition rates, we explicitly calculate $g$ defined in \eqref{eq:g} in Appendix~\ref{sec:generalised-amplitude-damping}.
Together with Theorem~\ref{thm:letter-identity}, this gives
\begin{equation*}
 g=\frac{|\mu_+-\mu_-|}{4},\qquad
 T_0=\frac{2\pi}{|\mu_+-\mu_-|}.
\end{equation*}

At this total exposure time, the type-II error can be made arbitrarily small while keeping the type-I error exactly zero.

\section{Discussion and outlook}\label{sec:letter-discussion}

We have studied the limits of discriminating Lindbladian dynamics with adaptive control, focusing on the roles of total exposure time and temporal resolution.
We characterised the Stein exponent as the largest instantaneous rate of increase in relative entropy and gave an algebraic criterion for its finiteness.
Its divergence has a direct finite-time consequence: at fixed positive type-I tolerance, the type-II error can be made arbitrarily small at sufficiently large finite exposure time.
We distinguished whether this is achieved by a finite protocol with zero type-II error or requires increasingly fine Zeno control, and quantified the gain from shorter resolution times.
For an identity null hypothesis, we also determined the exact threshold time, including at zero type-I error.
Together, these results establish temporal resolution as a resource alongside total exposure time, including for purely dissipative evolution.

A central open problem is to sharpen \eqref{eq:resolutionlaw} beyond its scaling order.
Under the assumptions of Theorem~\ref{thm:letter-resolution}, does the ratio of $-\log\bmin(\cL_1\|\cL_2)$ to $t\log(s_0/\smin)$ converge to a quantity determined by the generators alone, in the joint limit $t\to\infty$, $\smin\to0$ at every fixed $0<\eps<1$?
A characterisation of this limit in terms of the jump operators and drifts would determine the optimal quantitative tradeoff between exposure time and temporal resolution.

A complementary challenge is to determine how this tradeoff changes when controls and measurements take time and introduce errors.
In particular, what intervention rate minimises the discrimination error at a fixed total elapsed time, including control and readout?
Resolving this question would identify the conditions under which the temporal-resolution enhancement can be realised experimentally.

\comment{\section{Discussion}
Lindbladian discrimination combines the statistical choice of a tester with the dynamical choice of its call durations.
The differential characterisation determines the Stein exponent, while the generator-span criterion decides its finiteness.
Keeping the resolution time explicit reveals how an infinite Stein exponent is achieved: by a finite support test, or only through increasingly fine interventions.
Generalised amplitude damping shows that the latter distinction has a direct physical consequence, including a vanishing symmetric error at finite total time.
The exact threshold for general pairs of generators and the leading coefficient of the logarithmic law for the resolution time remain open, as does the extent of the advantage when controls have a finite duration or strength.
}

\renewcommand{\bibsection}{\section*{References}}
\bibliographystyle{apsrev4-2-notitles}
\bibliography{references}

\onecolumngrid
\clearpage
\makeatletter
\let\section\letter@section
\let\subsection\letter@subsection
\let\normalsize\letter@normalsize
\normalsize
\skip\footins=\letter@footinskip\relax
\interfootnotelinepenalty=\letter@interfootnotelinepenalty\relax
\let\raggedcolumn@sw\letter@raggedcolumn@sw
\makeatother
\setlength{\parskip}{3pt}
\setlength{\parindent}{1.2em}
\setlist{nosep}
\section*{Acknowledgements}
LL acknowledges financial support from the European Union under the European Research Council (ERC
Grant Agreement No. 101165230). 
MBP acknowledges support by the European Research Council (ERC Synerg Grant Agreement No. 856432) and the QuantERA Grant Agreement No 100726 (CoMaQLev). 
RS acknowledges support by the European Research Council (ERC Grant Agreement No. 948139 and ERC Grant AlgoQIP, Agreement No. 851716), from the Excellence Cluster Matter and Light for Quantum Computing (ML4Q-2), from the QuantERA II Programme of the European Union’s Horizon 2020 research and innovation programme under Grant Agreement No 101017733 (VERIqTAS) as well as the government grant managed by the Agence Nationale de la Recherche under the Plan France 2030 with the reference ANR-22-PETQ-0007.
\section*{AI declaration}
The authors formulated the central research question, derived the characterisation of the Stein exponent in Theorem~\ref{thm:letter-rate}, and evaluated it for the depolarising semigroup.
The authors further studied generalised amplitude damping with unequal positive transition rates, finding an infinite Stein exponent despite Choi support inclusion and observing that finer temporal resolution improves distinguishability.
ChatGPT 5.5, 5.6 Sol and 6 Astra were then used for further mathematical exploration of this phenomenon.
In particular, these AI models assisted in finding proofs of the Stein exponent dichotomy, the resolution-time scaling of the optimal type-II error probability, and the exact threshold for an identity null hypothesis.
AI also assisted with writing, proofreading and improving the exposition.
The authors have verified all AI-generated material and take responsibility for the mathematical claims and final presentation.
\par\medskip

\setcounter{secnumdepth}{3}
\setcounter{section}{0}
\setcounter{equation}{0}
\setcounter{theorem}{0}
\setcounter{proposition}{0}
\setcounter{lemma}{0}
\setcounter{corollary}{0}
\setcounter{figure}{0}
\setcounter{table}{0}
\renewcommand{\thesection}{A\arabic{section}}
\renewcommand{\thesubsection}{\thesection.\arabic{subsection}}
\renewcommand{\thesubsubsection}{\thesubsection.\arabic{subsubsection}}
\renewcommand{\theequation}{A\arabic{equation}}
\renewcommand{\thetheorem}{A\arabic{theorem}}
\renewcommand{\theproposition}{A\arabic{proposition}}
\renewcommand{\thelemma}{A\arabic{lemma}}
\renewcommand{\thecorollary}{A\arabic{corollary}}
\renewcommand{\thefigure}{A\arabic{figure}}
\renewcommand{\thetable}{A\arabic{table}}

\renewcommand{\theHsection}{appendix.\arabic{section}}
\renewcommand{\theHsubsection}{\theHsection.\arabic{subsection}}
\renewcommand{\theHsubsubsection}{\theHsubsection.\arabic{subsubsection}}
\renewcommand{\theHequation}{appendix.\arabic{equation}}
\renewcommand{\theHtheorem}{appendix.\arabic{theorem}}
\renewcommand{\theHproposition}{appendix.\arabic{proposition}}
\renewcommand{\theHlemma}{appendix.\arabic{lemma}}
\renewcommand{\theHcorollary}{appendix.\arabic{corollary}}
\renewcommand{\theHfigure}{appendix.\arabic{figure}}
\renewcommand{\theHtable}{appendix.\arabic{table}}

\makeatletter
\renewcommand{\p@subsection}{}
\renewcommand{\p@subsubsection}{}
\let\l@abstract\@gobbletwo
\let\appendix@section\l@section
\let\appendix@subsection\l@subsection
\let\appendix@subsubsection\l@subsubsection
\let\l@section\@gobbletwo
\let\l@subsection\@gobbletwo
\let\l@subsubsection\@gobbletwo
\renewcommand{\tocdepth@restore}{%
  \let\l@section\appendix@section
  \let\l@subsection\appendix@subsection
  \let\l@subsubsection\appendix@subsubsection
}
\makeatother
\pdfbookmark[0]{Appendix}{appendix}
\begin{center}
 {\large\bfseries Appendix}
\end{center}

\tableofcontents
\medskip

\section{Preliminaries}\label{sec:model}

The quantum relative entropy~\cite{Umegaki1962} is given by $$D(\rho\|\sigma)=\Tr[\rho(\log\rho-\log\sigma)]$$ for $\supp\rho\subseteq\supp\sigma$ and  $D(\rho\|\sigma)=\infty$ otherwise.
For $0\le\eps<1$, the hypothesis-testing relative entropy~\cite{WangRenner2012} is
\begin{equation*}
 D_H^\eps(\rho\|\sigma)\coloneqq
 -\log\inf_{\substack{0\le Q\le\1\\\Tr(Q\rho)\ge1-\eps}}\Tr(Q\sigma).
\end{equation*}
Thus $2^{-D_H^\eps(\rho\|\sigma)}$ is the optimal type-II error probability for null hypothesis $\rho$ and alternative hypothesis $\sigma$, subject to type-I error at most $\eps$.
We write $h(x)=-x\log x-(1-x)\log(1-x)$ for binary entropy, with $0\log0=0$.

A state divergence $\mathbf D$ satisfying data processing extends to channels by maximising the output divergence over a common input, including an auxiliary system~\cite{Wilde2020}.
The corresponding channel divergence is then given
\begin{equation*}
 \mathbf D(\cN\|\cM)\coloneqq
 \sup_{E,\omega_{AE}}
 \mathbf D\left((\cN\otimes\id_E)(\omega_{AE})
 \middle\|(\cM\otimes\id_E)(\omega_{AE})\right),
\end{equation*}
for $\cN,\cM$ being quantum channel $A$ to $B$ and 
where the supremum ranges over normalised states and finite-dimensional auxiliary systems $E$.
We use this extension for both $D$ and $D_H^\eps$.
In particular, $D_H^\eps(\cN^{\otimes n}\|\cM^{\otimes n})$ is the negative logarithm of the optimal type-II error probability for $n$ parallel channel uses, allowing an entangled input across all uses and the auxiliary system, followed by a joint binary measurement.

To quantify the increase in distinguishability when the two inputs may already differ, define the amortised relative entropy~\cite{Wilde2020} by
\begin{equation}
 \DA(\cN\|\cM)=
 \sup_{\substack{E,\rho_{AE},\sigma_{AE}\\D(\rho_{AE}\|\sigma_{AE})<\infty}}
 \left\{D((\cN\otimes\id_E)(\rho_{AE})\|(\cM\otimes\id_E)(\sigma_{AE}))
                -D(\rho_{AE}\|\sigma_{AE})\right\}.
 \label{sm:amort}
\end{equation}
The supremum includes normalised states on $AE$ and all finite-dimensional auxiliary systems $E$.
It measures the largest output distinguishability after subtracting the distinguishability already present at the input.
We also use the regularised channel relative entropy~\cite{Wilde2020}, defined by taking the supremum of the channel relative entropy per use over tensor powers.
The chain rule for quantum relative entropy~\cite{Fang2020} gives
\begin{equation}
 \DA(\cN\|\cM)=\Dreg(\cN\|\cM)
 =\sup_{k\ge1}\frac1k
   D\left(\cN^{\otimes k}\middle\|\cM^{\otimes k}\right).
 \label{sm:chain}
\end{equation}
Here $k$ is a positive integer, and the channel relative entropy includes the optimisation over states on $A^kE$ described above.
Purifying the common input and using data processing shows that this optimisation may be restricted to pure states.
All auxiliary spaces in each individual optimisation are finite.

Let $A'\simeq A$ be a reference and write $\ket{\Phi}=d^{-1/2}\sum_{a=1}^d\ket a_A\ket a_{A'}$ on $AA'$.
The normalised Choi state is $J_{BA'}(\cN)\coloneqq(\cN\otimes\id_{A'})\left(\kb{\Phi}\right)$, where $\kb{\Phi}=\ket{\Phi}\!\bra{\Phi}$.
For the semigroups considered here, $B=A$, and denote $J^{(i)}_{AA'}(s)\coloneqq J_{AA'}\left(e^{s\cL_i}\right)$.
We write $\cN\cpleq\cM$ when $\cM-\cN$ is completely positive.
In finite dimension,
\begin{equation}
 \supp J_{BA'}(\cN)\subseteq\supp J_{BA'}(\cM)
 \quad\Longleftrightarrow\quad
 \cN\cpleq c\cM\quad\text{for some }c<\infty.
 \label{sm:cporder}
\end{equation}
Indeed, the inclusion is equivalent to a finite bound between the two positive Choi matrices on the support of $J_{BA'}(\cM)$.
The same order holds after tensoring with any identity map.

\section{The Stein exponent and its differential characterisation}\label{sec:rate}

In the following, we give the proof of Theorem~\ref{thm:letter-rate}.
For a positive resolution time $\smin>0$, define
\begin{equation*}
 \Gamma_{\smin}(\cL_1\|\cL_2)\coloneqq
 \sup_{s\ge\smin}\frac1s\DA\left(e^{s\cL_1}\middle\|e^{s\cL_2}\right),
\end{equation*}
where the supremum may equal $+\infty$.

\begin{theorem}\label{thm:rate}
The Lindbladian Stein exponent satisfies
\begin{align}
 \nonumber \gamma(\cL_1\|\cL_2)
 &=\sup_{s>0}\frac{\DA\left(e^{s\cL_1}\middle\|e^{s\cL_2}\right)}s
   =\lim_{s\to0}\frac{\DA\left(e^{s\cL_1}\middle\|e^{s\cL_2}\right)}s \\
 &=\sup_{E,\rho_{AE},\sigma_{AE}>0}
   \left.\frac{\dd}{\dd s}
       D\left(\left(e^{s\cL_1}\otimes\id_E\right)(\rho_{AE})\|
         \left(e^{s\cL_2}\otimes\id_E\right)(\sigma_{AE})\right)\right|_{s=0}.
 \label{sm:diff}
\end{align}
The supremum in \eqref{sm:diff} ranges over finite-dimensional auxiliary systems $E$ and positive definite states $\rho_{AE},\sigma_{AE}$ on $AE$.
Moreover, replacing the inner $\liminf_{t\to\infty}$ in the definition of $\gamma(\cL_1\|\cL_2)$ by $\limsup_{t\to\infty}$ gives the same result.
Furthermore, for every fixed $\smin>0$, the optimal type-II error probability $\bmin(\cL_1\|\cL_2)$ defined in \eqref{eq:beta} satisfies
\begin{align*}
 \Gamma_{\smin}(\cL_1\|\cL_2)
 &=\lim_{\eps\to0}\liminf_{t\to\infty}
 -\frac1t\log\bmin(\cL_1\|\cL_2)=\lim_{\eps\to0}\limsup_{t\to\infty}
 -\frac1t\log\bmin(\cL_1\|\cL_2).
\end{align*}
Moreover, whenever $\Gamma_{\smin}(\cL_1\|\cL_2)<\infty$, we have, for all $0\le\eps<1$,
\begin{equation}
 \bmin(\cL_1\|\cL_2)\ge
 2^{-[t\Gamma_{\smin}(\cL_1\|\cL_2)+h(\eps)]/(1-\eps)}.
 \label{sm:telescopebound}
\end{equation}
The same bound with $\gamma(\cL_1\|\cL_2)$ holds for $\bet(\cL_1\|\cL_2)$ whenever $\gamma(\cL_1\|\cL_2)<\infty$.
\end{theorem}

\begin{proof}
We first establish achievability using parallel strategies, which are a special case of the adaptive protocols considered here.
Fix $s>0$ and prepare a state on $A^nE$, involving $n$ probes and an auxiliary system $E$.
Apply the unknown dynamics to each probe for time $s$, followed by a joint binary measurement.
This applies $(e^{s\cL_i})^{\otimes n}\otimes\id_E$ under hypothesis $i$ and uses total evolution time $ns$.
It can be implemented sequentially by storing the remaining probes and all outputs in the auxiliary system and swapping each probe into $A$ for its evolution period~\cite{Bergh2024}.

The negative logarithm of the optimal type-II error probability for this strategy is the hypothesis-testing relative entropy of the tensor-power channels.
Its asymptotic rate is the regularised channel relative entropy~\cite{WangWilde2019Channels}:
\begin{equation*}
 \lim_{\eps\to0}\liminf_{n\to\infty}
 \frac1n D_H^\eps\left(\cN^{\otimes n}\middle\|\cM^{\otimes n}\right)
 =\Dreg(\cN\|\cM).
\end{equation*}
For a total time budget $t$, choose $n=\lfloor t/s\rfloor$.
Since these parallel strategies are admissible adaptive protocols,
\begin{equation*}
 -\log\bet(\cL_1\|\cL_2)
 \ge  D_H^\eps\left(\left(e^{s\cL_1}\right)^{\otimes n}
 \middle\|\left(e^{s\cL_2}\right)^{\otimes n}\right),
\end{equation*}
Dividing by $t$ and taking limits, using $1/s-1/t<n/t\le1/s$ and \eqref{sm:chain}, gives
\begin{align*}
 \gamma(\cL_1\|\cL_2)
 &=\lim_{\eps\to0}\liminf_{t\to\infty}
 -\frac1t\log\bet(\cL_1\|\cL_2)\\
 &\ge\lim_{\eps\to0}\liminf_{t\to\infty}
 \frac nt\,\frac1n D_H^\eps\left(\left(e^{s\cL_1}\right)^{\otimes n}
 \middle\|\left(e^{s\cL_2}\right)^{\otimes n}\right)\\
 &=\frac1s\Dreg\left(e^{s\cL_1}\middle\|e^{s\cL_2}\right)
 =\frac1s\DA\left(e^{s\cL_1}\middle\|e^{s\cL_2}\right).
\end{align*}
Taking the supremum over $s>0$ gives the required lower bound.
At fixed positive resolution time, the same argument with $s\ge\smin$ gives the lower bound $\Gamma_{\smin}(\cL_1\|\cL_2)$.

For the converse, first consider a protocol with fixed durations $s_1,\ldots,s_n$ and common control channels $\cA_k$ on $AE$.
Starting from $\rho_{0,AE}^{(1)}=\rho_{0,AE}^{(2)}=\omega_{AE}$, the states after the $k$th evolution and control operation are, for $i=1,2$,
\begin{equation*}
 \rho_{k,AE}^{(i)}
 =\cA_k\left(\left(e^{s_k\cL_i}\otimes\id_E\right)
 (\rho_{k-1,AE}^{(i)})\right).
\end{equation*}
Data processing through $\cA_k$ and the definition of amortised relative entropy give
\begin{align*}
 D(\rho_{k,AE}^{(1)}\|\rho_{k,AE}^{(2)})
 &\le
 D\left(\left(e^{s_k\cL_1}\otimes\id_E\right)(\rho_{k-1,AE}^{(1)})
 \middle\|\left(e^{s_k\cL_2}\otimes\id_E\right)(\rho_{k-1,AE}^{(2)})\right)
 \\
 &\le D(\rho_{k-1,AE}^{(1)}\|\rho_{k-1,AE}^{(2)})
 +\DA\left(e^{s_k\cL_1}\middle\|e^{s_k\cL_2}\right).
\end{align*}
If the preceding relative entropy is infinite, the second inequality is trivial. Otherwise, it follows directly from the definition of $\DA$.
The initial relative entropy is zero, so iterating gives
\begin{equation*}
 D(\rho_{n,AE}^{(1)}\|\rho_{n,AE}^{(2)})
 \le\sum_{k=1}^n\DA\left(e^{s_k\cL_1}\middle\|e^{s_k\cL_2}\right).
\end{equation*}

The same argument applies when the evolution durations $s_k$ and the number of calls are chosen adaptively based on preceding measurement outcomes.
At each stage, retain the preceding classical history $h$ in a separate classical register $C$ and write its probabilities as $p_i(h)$ and its conditional states as $\rho_{h,AE}^{(i)}$.
The direct-sum states act on $CAE$, and their relative entropy is
\begin{equation*}
 D\!\left(\bigoplus_h p_1(h)\rho_{h,AE}^{(1)}\,\middle\|\,
           \bigoplus_h p_2(h)\rho_{h,AE}^{(2)}\right)
 =D(p_1\|p_2)+\sum_h p_1(h)D(\rho_{h,AE}^{(1)}\|\rho_{h,AE}^{(2)}).
\end{equation*}
Applying the same additive bound to the conditional states bounds the relative entropy after a history-dependent call by its value before the call plus
$\sum_h p_1(h)\DA\left(e^{s_h\cL_1}\middle\|e^{s_h\cL_2}\right)$.
Each subsequent instrument is a common channel on the state and record and cannot increase the total relative entropy.
Iterating therefore gives an expectation under the null hypothesis, denoted by $\mathbb E_1$.
Since the total time is at most $t$ on every branch, for $t>0$ we obtain
\begin{align*}
 D(\rho_{AE}^{(1),\mathsf P}\|\rho_{AE}^{(2),\mathsf P})
 &\le\mathbb E_1\!\left[\sum_k\DA\left(e^{s_k\cL_1}\middle\|e^{s_k\cL_2}\right)\right]\le t\sup_{s>0}\frac{\DA\left(e^{s\cL_1}\middle\|e^{s\cL_2}\right)}s.
\end{align*}
Nodes with zero probability under the null hypothesis make no contribution.
At fixed positive resolution time, the supremum is restricted to $s\ge\smin$ and hence equals $\Gamma_{\smin}(\cL_1\|\cL_2)$.

To turn this entropy bound into an error bound, use the relation between hypothesis-testing and quantum relative entropy~\cite{WangRenner2012}.
For every protocol $\mathsf P$ with $\alpha(\mathsf P)\le\eps$,
\begin{align*}
 -\log\beta(\mathsf P)
 &\le D_H^\eps\left(\rho_{AE}^{(1),\mathsf P}\middle\|\rho_{AE}^{(2),\mathsf P}\right)
 \le\frac{D(\rho_{AE}^{(1),\mathsf P}\|\rho_{AE}^{(2),\mathsf P})+h(\eps)}{1-\eps}
 \\
 &\le\frac{t\displaystyle\sup_{s>0}\frac1s\DA\left(e^{s\cL_1}\middle\|e^{s\cL_2}\right)+h(\eps)}{1-\eps}.
\end{align*}
Optimising over all admissible protocols gives the corresponding lower bound on the optimal type-II error probability, with the supremum restricted to $s\ge\smin$ for \eqref{sm:telescopebound}.
Dividing by $t$, taking $\limsup_{t\to\infty}$ and then $\eps\to0$ gives the upper bound $\sup_{s>0}\DA\left(e^{s\cL_1}\middle\|e^{s\cL_2}\right)/s$ on the iterated exponent, or $\Gamma_{\smin}(\cL_1\|\cL_2)$ at fixed positive resolution time.
Together with achievability, this proves the operational equality with $\sup_{s>0}\DA\left(e^{s\cL_1}\middle\|e^{s\cL_2}\right)/s$, the fixed-resolution statement, and the equality of the iterated limits using $\liminf$ or $\limsup$.
If the corresponding supremum is infinite, achievability already proves these assertions.
Finally, Lemma~\ref{lem:amort-differential} below identifies the unrestricted supremum with the small-time limit and the differential in \eqref{sm:diff}, completing the proof.
\end{proof}

The following lemma supplies the differential step by using the semigroup property and the fundamental theorem of calculus.
\begin{lemma}\label{lem:amort-differential}
For two finite-dimensional Lindbladians,
\begin{align}
 \sup_{s>0}\frac{\DA\left(e^{s\cL_1}\middle\|e^{s\cL_2}\right)}s
 &=\lim_{s\to0}\frac{\DA\left(e^{s\cL_1}\middle\|e^{s\cL_2}\right)}s\nonumber\\
 &=\sup_{E,\rho_{AE},\sigma_{AE}>0}\left.\frac{\dd}{\dd r}
 D\left(\left(e^{r\cL_1}\otimes\id_E\right)(\rho_{AE})
 \middle\|\left(e^{r\cL_2}\otimes\id_E\right)(\sigma_{AE})\right)\right|_{r=0}.
 \label{sm:amort-differential}
\end{align}
The final supremum ranges over positive definite states and finite-dimensional auxiliary systems, as in \eqref{sm:diff}.
\end{lemma}

\begin{proof}
In \eqref{sm:amort}, it suffices to use positive definite input pairs.
To see this, replace each input by $(1-\delta)\rho_{AE}+\delta\pi_{AE}$ or $(1-\delta)\sigma_{AE}+\delta\pi_{AE}$, where $\pi_{AE}$ is maximally mixed on the joint input space.
For $D(\rho_{AE}\|\sigma_{AE})<\infty$, the input relative entropy converges to its original value.
Lower semicontinuity gives the required lower bound on the limit inferior of the output relative entropy.
Hence regularisation can approximate each entropy increment from below.

For each such pair, the definition of amortised channel relative entropy gives
\begin{align*}
 \liminf_{s\to0}\frac{\DA\left(e^{s\cL_1}\middle\|e^{s\cL_2}\right)}s
 &\ge\lim_{s\to0}\frac1s\Bigl[
 D\left(\left(e^{s\cL_1}\otimes\id_E\right)(\rho_{AE})
 \middle\|\left(e^{s\cL_2}\otimes\id_E\right)(\sigma_{AE})\right)
 -D(\rho_{AE}\|\sigma_{AE})\Bigr]\\
 &=\left.\frac{\dd}{\dd s}
 D\left(\left(e^{s\cL_1}\otimes\id_E\right)(\rho_{AE})
 \middle\|\left(e^{s\cL_2}\otimes\id_E\right)(\sigma_{AE})\right)\right|_{s=0}.
\end{align*}
Taking the supremum over these pairs gives the required lower bound on the limit inferior.

For the upper bound, assume that the supremum of the derivatives in \eqref{sm:amort-differential} is finite.
A positive definite input stays positive definite at finite time: the channel contains the invertible no-jump map $X\mapsto e^{sG_i}Xe^{sG_i^*}$.
Write $\rho_{AE}(u)=(e^{u\cL_1}\otimes\id_E)(\rho_{AE})$ and $\sigma_{AE}(u)=(e^{u\cL_2}\otimes\id_E)(\sigma_{AE})$.
By the semigroup property, the derivative at time $u$ is the derivative at zero starting from this evolved pair, so
\begin{align*}
 \frac{\dd}{\dd u}D(\rho_{AE}(u)\|\sigma_{AE}(u))
 &=\left.\frac{\dd}{\dd r}
 D\left(\left(e^{r\cL_1}\otimes\id_E\right)(\rho_{AE}(u))
 \middle\|\left(e^{r\cL_2}\otimes\id_E\right)(\sigma_{AE}(u))\right)\right|_{r=0}\\
 &\le\sup_{E,\widetilde\rho_{AE},\widetilde\sigma_{AE}>0}\left.\frac{\dd}{\dd r}
 D\left(\left(e^{r\cL_1}\otimes\id_E\right)(\widetilde\rho_{AE})
 \middle\|\left(e^{r\cL_2}\otimes\id_E\right)(\widetilde\sigma_{AE})\right)\right|_{r=0}.
\end{align*}
The fundamental theorem of calculus now gives
\begin{align*}
 D(\rho_{AE}(s)\|\sigma_{AE}(s))-D(\rho_{AE}\|\sigma_{AE})
 &=\int_0^s\frac{\dd}{\dd u}D(\rho_{AE}(u)\|\sigma_{AE}(u))\,\dd u\\
 &\le s\sup_{E,\widetilde\rho_{AE},\widetilde\sigma_{AE}>0}\left.\frac{\dd}{\dd r}
 D\left(\left(e^{r\cL_1}\otimes\id_E\right)(\widetilde\rho_{AE})
 \middle\|\left(e^{r\cL_2}\otimes\id_E\right)(\widetilde\sigma_{AE})\right)\right|_{r=0}.
\end{align*}
Taking the supremum over the input pairs, using the regularisation argument and combining with the lower bound yields
\begin{align*}
 &\sup_{E,\rho_{AE},\sigma_{AE}>0}\left.\frac{\dd}{\dd r}
 D\left(\left(e^{r\cL_1}\otimes\id_E\right)(\rho_{AE})
 \middle\|\left(e^{r\cL_2}\otimes\id_E\right)(\sigma_{AE})\right)\right|_{r=0}\\
 &\quad\le\liminf_{s\to0}\frac{\DA\left(e^{s\cL_1}\middle\|e^{s\cL_2}\right)}s
 \le\limsup_{s\to0}\frac{\DA\left(e^{s\cL_1}\middle\|e^{s\cL_2}\right)}s
 \le\sup_{s>0}\frac{\DA\left(e^{s\cL_1}\middle\|e^{s\cL_2}\right)}s\\
 &\quad\le\sup_{E,\rho_{AE},\sigma_{AE}>0}\left.\frac{\dd}{\dd r}
 D\left(\left(e^{r\cL_1}\otimes\id_E\right)(\rho_{AE})
 \middle\|\left(e^{r\cL_2}\otimes\id_E\right)(\sigma_{AE})\right)\right|_{r=0}.
\end{align*}
If the supremum of the derivatives is infinite, the lower bound on the limit inferior already proves \eqref{sm:amort-differential}.
\end{proof}

We next give an explicit expression for the relative entropy derivative.
For a positive definite operator $C$ and an operator $Z$ on the same space, the derivative of the base-two logarithm has the integral representation~\cite{AudenaertEisert2005}
\begin{equation}
 \Omega_C(Z)=(\log e)\int_0^\infty
       (C+u\1)^{-1}Z(C+u\1)^{-1}\,\dd u .
 \label{sm:omega}
\end{equation}
If $C=\sum_a c_a\kb a$ is its spectral decomposition, then the coefficient multiplying the matrix element $Z_{ab}=\bra aZ\ket b$ is $(\log c_a-\log c_b)/(c_a-c_b)$, with continuous value $(\log e)/c_a$ when $c_a=c_b$.
For positive definite states $\rho_{AE},\sigma_{AE}$ on $AE$, differentiating the entropy and using trace preservation gives
\begin{align}
 &\left.\frac{\dd}{\dd s}D\!\left(\left(e^{s\cL_1}\otimes\id_E\right)(\rho_{AE})
 \middle\|\left(e^{s\cL_2}\otimes\id_E\right)(\sigma_{AE})\right)\right|_{s=0}
 \nonumber\\
 &\quad=\Tr\!\left[(\cL_1\otimes\id_E)(\rho_{AE})(\log\rho_{AE}-\log\sigma_{AE})\right]
 -\Tr\!\left[\rho_{AE}\,\Omega_{\sigma_{AE}}((\cL_2\otimes\id_E)(\sigma_{AE}))\right].
 \label{sm:explicitdiff}
\end{align}
The factor $\log e$ in \eqref{sm:omega} is required by the logarithm convention.
\eqref{sm:explicitdiff} is evaluated on positive definite inputs. No differentiation at a singular state is needed in Theorem~\ref{thm:rate}.

We can also bound longer evolution intervals by subdividing them.
Inserting the relative entropy after the first pair of channels gives
\begin{equation*}
 \DA\left(e^{(u+v)\cL_1}\middle\|e^{(u+v)\cL_2}\right)
 \le\DA\left(e^{u\cL_1}\middle\|e^{u\cL_2}\right)
 +\DA\left(e^{v\cL_1}\middle\|e^{v\cL_2}\right).
\end{equation*}
For $s\ge\smin$, take $m=\lfloor s/\smin\rfloor$ and $u=s/m$.
Then $\smin\le u<2\smin$ and
\begin{equation*}
 \DA\left(e^{s\cL_1}\middle\|e^{s\cL_2}\right)
 \le m\DA\left(e^{u\cL_1}\middle\|e^{u\cL_2}\right),
\end{equation*}
so
\begin{equation}
 \Gamma_{\smin}(\cL_1\|\cL_2)
 =\sup_{\smin\le u<2\smin}\frac{\DA\left(e^{u\cL_1}\middle\|e^{u\cL_2}\right)}u .
 \label{sm:subdivision}
\end{equation}
This is an identity for channel divergences. It does not require the tester to introduce any forbidden shorter call.

\section{A criterion for a finite Stein exponent}\label{sec:finite}

We now prove the algebraic criterion in Theorem~\ref{thm:letter-span}, showing that finiteness of the Stein exponent can be tested on the jump operators and drift alone.
For the sufficient direction, we construct a common drift and bound the explicit differential in \eqref{sm:explicitdiff}.
The necessary direction follows from the discrimination protocol in Section~\ref{sec:zeno}.

Consider two Lindblad generators on a finite-dimensional system $A$, given for $i=1,2$ by
\begin{equation*}
 \cL_i(X)=G_iX+XG_i^*+\sum_jL_j^{(i)}XL_j^{(i)*},
\end{equation*}
with $G_i+G_i^*+\sum_jL_j^{(i)*}L_j^{(i)}=0$.

\begin{theorem}[Finite Stein exponent]\label{thm:span}
Let $\cS_2=\Span\{\1,L_j^{(2)}:j\}$.
Then
\begin{equation}
 \gamma(\cL_1\|\cL_2)<\infty
 \quad\Longleftrightarrow\quad
 L_j^{(1)}\in\cS_2\ \text{for all }j
 \quad\text{and}\quad G_1-G_2\in\cS_2 .
 \label{sm:span}
\end{equation}
Furthermore, if these conditions hold and $\cL_1\ne\cL_2$, then for all fixed $0<\eps<1$,
$-\log\bet(\cL_1\|\cL_2)=\Theta(t)$ as $t\to\infty$.
The same statement holds at any fixed positive resolution time, i.e.~$-\log\bmin(\cL_1\|\cL_2)=\Theta(t)$ as $t\to\infty$ for fixed $\smin>0$ and $0<\eps<1$.
\end{theorem}

In this section we prove the implication $\Leftarrow$ in \eqref{sm:span}, i.e.~that $\gamma(\cL_1\|\cL_2)<\infty$ whenever $L_j^{(1)}\in\cS_2$ for all $j$ and $G_1-G_2\in\cS_2$.
For this, we use a common-drift representation to bound the relative-entropy increase in a single small time step uniformly over the input states and auxiliary systems.
The implication $\Rightarrow$ in \eqref{sm:span} follows from Theorem~\ref{thm:infinite} in Section~\ref{sec:zeno} below.

\begin{proof}
Assume the span conditions in \eqref{sm:span}.
Without loss of generality, we may choose representations with a nonzero identity jump under each hypothesis and a common drift,\footnote{To see this, write $G_1-G_2=a_0\1+\sum_k b_kL_k^{(2)}$ and change only the representation under the alternative hypothesis according to
\begin{align}
 L_k^{(2)}&\longmapsto L_k^{(2)}-\overline b_k\1,\nonumber\\
 G_2&\longmapsto G_2+\sum_k b_kL_k^{(2)}-\tfrac12\sum_k|b_k|^2\1.
 \label{sm:gauge}
\end{align}
Direct expansion shows that $\cL_2$ is unchanged, and $\Span\{\1,L_k^{(2)}-\overline b_k\1:k\}=\cS_2$.
The new drift under the alternative hypothesis is $G_1-(a_0+\frac12\sum_k|b_k|^2)\1$.
Removing its irrelevant imaginary scalar gives $G_1-a\1$, where $a=\operatorname{Re}a_0+\frac12\sum_k|b_k|^2$. The drift under the null hypothesis is still $G_1$.
Choose $r>\max\{0,a\}$ and add identity jumps $\sqrt{2r}\1$ and $\sqrt{2(r-a)}\1$ under the null hypothesis and alternative hypothesis, respectively.
Subtracting $r\1$ from the drift under the null hypothesis and $(r-a)\1$ from the drift under the alternative hypothesis preserves both generators and makes their drifts equal to $G_1-r\1$.
These changes preserve both augmented jump spans.}
\begin{equation*}
 G_1=G_2=G.
\end{equation*}
Let $\cJ_1,\cJ_2$ be the completely positive maps obtained by summing the jump terms in these representations, i.e. for $i=1,2$,
\begin{equation*}
 \cJ_i(X)\coloneqq\sum_j L_j^{(i)}XL_j^{(i)*}.
\end{equation*}
The Kraus operators of $\cJ_2$ span $\cS_2$, since its identity jump is nonzero and its other Kraus operators are scalar shifts of the original jumps under the alternative hypothesis.
Every Kraus operator of $\cJ_1$ lies in this span by assumption.
Thus the Choi support of $\cJ_1$ is contained in that of $\cJ_2$, and finite-dimensional positive-operator comparison gives $\cJ_1\cpleq M\cJ_2$ for some finite $M\ge1$.
Trace preservation then gives
\begin{equation*}
 \cL_i(X)=GX+XG^*+\cJ_i(X),\qquad
 K:=\cJ_1^*(\1)=\cJ_2^*(\1)=-G-G^*,
\end{equation*}
Here $K>0$ because the representations contain nonzero identity jumps. Set $\kappa=\op{K}$.
We can hence write $G=-iH-K/2$ for some self-adjoint operator $H$ and for $0<\delta<1/\kappa$ define
\begin{equation*}
 V_\delta=e^{-i\delta H}\sqrt{\1-\delta K},\qquad
 \mathcal F_{i,\delta}(X)=V_\delta XV_\delta^*\oplus\delta\cJ_i(X),\qquad
 \cN_{i,\delta}(X)=V_\delta XV_\delta^*+\delta\cJ_i(X).
\end{equation*}
Since $V_\delta^*V_\delta+\delta\cJ_i^*(\1)=\1$, both maps are quantum channels, and $\cN_{i,\delta}$ is obtained from $\mathcal F_{i,\delta}$ by discarding the two-valued classical flag.
Moreover, $V_\delta=\1+\delta G+O(\delta^2)$, so
\begin{equation*}
 \cN_{i,\delta}=\id+\delta\cL_i+O(\delta^2).
\end{equation*}

Fix positive definite states $\rho_{AE},\sigma_{AE}$ and any finite-dimensional auxiliary system $E$.
For the following calculation, we omit the state subscripts and understand all maps on $A$ to act trivially on $E$.
Since $\mathcal F_{1,\delta}(\sigma)$ and $\mathcal F_{2,\delta}(\sigma)$ have the same first block and $\cJ_2(\sigma)\ge M^{-1}\cJ_1(\sigma)$, we obtain
\begin{align*}
 D(\cN_{1,\delta}(\rho)\|\cN_{2,\delta}(\sigma))
 &\le D(\mathcal F_{1,\delta}(\rho)\|\mathcal F_{2,\delta}(\sigma))\le D(\mathcal F_{1,\delta}(\rho)\|\mathcal F_{1,\delta}(\sigma))
       +\delta\log M\,\Tr\cJ_1(\rho)\le D(\rho\|\sigma)+\delta\kappa\log M,
\end{align*}
where the first inequality follows from data processing under discarding the flag, the second uses operator monotonicity of the logarithm, and the last combines data processing for $\mathcal F_{1,\delta}$ with $\Tr\cJ_1(\rho)=\Tr[(K\otimes\1_E)\rho]\le\kappa$.

Subtracting $D(\rho\|\sigma)$, dividing by $\delta$ and letting $\delta\to0$ now gives
\begin{align*}
 &\Tr[\cL_1(\rho)(\log\rho-\log\sigma)]
   -\Tr[\rho\,\Omega_\sigma(\cL_2(\sigma))]\\
 &\quad=\lim_{\delta\to0}
 \frac{D(\cN_{1,\delta}(\rho)\|\cN_{2,\delta}(\sigma))-D(\rho\|\sigma)}{\delta}
 \le\kappa\log M.
\end{align*}
Here we used differentiability at the positive definite input pair and the expansion of $\cN_{i,\delta}$, so the derivative is exactly \eqref{sm:explicitdiff}.
The bound is independent of the input pair and the auxiliary system.
Taking the supremum in Theorem~\ref{thm:letter-rate} therefore proves $\gamma(\cL_1\|\cL_2)\le\kappa\log M<\infty$.

We now prove the further claim that $-\log\bet(\cL_1\|\cL_2)=\Theta(t)$ as $t\to\infty$ for distinct generators and fixed $0<\eps<1$, together with its counterpart at fixed positive resolution time.
For the lower bound on the optimal type-II error probability, Theorem~\ref{thm:rate} gives, for $0\le\eps<1$ and $t\ge0$,
\begin{equation*}
 \bet(\cL_1\|\cL_2)\ge 2^{-[t\gamma(\cL_1\|\cL_2)+h(\eps)]/(1-\eps)}>0.
\end{equation*}

To prove the exponential upper bound on the optimal type-II error probability, fix $\smin>0$ and assume $\cL_1\ne\cL_2$.
Distinct generators have distinct channels for some $s\in(\smin,2\smin)$.
Indeed, if the channels agreed throughout this interval, differentiating as before and using invertibility of the exponentials would give
\begin{equation*}
 \cL_1=\left(\frac{\dd}{\dd s}e^{s\cL_1}\right)e^{-s\cL_1}
 =\left(\frac{\dd}{\dd s}e^{s\cL_2}\right)e^{-s\cL_2}=\cL_2,
\end{equation*}
contradicting $\cL_1\ne\cL_2$.
The corresponding Choi states are distinct, so applying the quantum Stein lemma~\cite{Hiai1991,Ogawa2000} to $m=\lfloor t/s\rfloor$ independent copies gives, for fixed $0<\eps<1$,
\begin{equation*}
 \liminf_{t\to\infty}-\frac1t\log\bmin(\cL_1\|\cL_2)
 \ge\frac1s D\left(J^{(1)}_{AA'}(s)\middle\|J^{(2)}_{AA'}(s)\right)>0.
\end{equation*}
Together with the preceding lower bound on the error, this proves the claimed $\Theta(t)$ behaviour.

For the necessary direction, Theorem~\ref{thm:infinite} in Section~\ref{sec:zeno} shows that failure of either span condition implies $\gamma(\cL_1\|\cL_2)=\infty$, completing the equivalence.
\end{proof}

\section{Repeated projections and amplification at finite total time}\label{sec:zeno}

In this section we show that, if either algebraic condition in \eqref{sm:span} fails, then $\bet(\cL_1\|\cL_2)=0$ for fixed $0<\eps<1$ and sufficiently large finite $t$.
To prove this, we construct adaptive protocols of the form shown in Figure~\ref{fig:protocol} satisfying $\alpha\le\eps$ and $\beta\to0$ as $\smin\to0$ at fixed total time $t$.
For the special case of $\cL_1=0$, we refine this construction in Section~\ref{sec:identity} to attain the exact threshold for vanishing optimal type-II error probability, including under the constraint $\alpha=0$.

\begin{theorem}[Vanishing optimal type-II error probability at finite total time]\label{thm:infinite}
If either condition in \eqref{sm:span} fails, then there exist constants $c,s_0,K>0$, depending only on the generators, such that, for all $0<\eps<1$, one can choose $0<t_\eps\le K[1+\log(1/\eps)]$ with
\begin{equation}
 \bmin(\cL_1\|\cL_2)\le\left(\frac{\smin}{2s_0}\right)^{ct},
 \label{sm:achievability}
\end{equation}
where $t\ge t_\eps$ and $0<\smin\le s_0$.
In particular, $\gamma(\cL_1\|\cL_2)=\infty$ and $\beta_{\eps,t}(\cL_1\|\cL_2)=0$ for all $t\ge t_\eps$.
Conversely, if for all $0<\eps<1$ the unrestricted optimal type-II error probability vanishes at some finite total time, then $\gamma(\cL_1\|\cL_2)=\infty$.
\end{theorem}

The unrestricted conclusions follow from \eqref{sm:achievability} by letting $\smin\to0$ and using Theorem~\ref{thm:rate}, which gives a strictly positive lower bound on the error whenever $\gamma(\cL_1\|\cL_2)<\infty$.
The proof constructs an explicit discrimination protocol using the quantum Zeno effect.
First, in Section~\ref{sec:zeno-gad-motivation}, we illustrate this protocol for the special case $\cL_1=0$ and $\cL_2=\cLgad{\mu_-}{\mu_+}$ with unequal positive transition rates.
We then establish the common projected evolution and its approximation by finite protocols in Section~\ref{sec:zeno-common}, which we use in Section~\ref{sec:zeno-general} to complete the proof of Theorem~\ref{thm:infinite}.

\subsection{An illustrative construction for generalised amplitude damping}\label{sec:zeno-gad-motivation}

We first illustrate the upper bound for the identity null hypothesis $\cL_1=0$ and the generalised amplitude damping alternative hypothesis $\cL_2=\cLgad{\mu_-}{\mu_+}$ defined in Section~\ref{sec:letter-gad}, with unequal positive transition rates.
The generator under the alternative hypothesis has jumps $L_-=\sqrt{\mu_-}\ket0\!\bra1$ and $L_+=\sqrt{\mu_+}\ket1\!\bra0$, and drift $G=-\tfrac12\left(\mu_+\kb{0}+\mu_-\kb{1}\right)$.
The protocol works as follows: we initialise the probe $A$ and an auxiliary qubit $E$ in $\kb{\psi}$, with maximally entangled state $\ket{\psi}=(\ket{00}+\ket{11})/\sqrt2$.
We evolve the probe under the unknown Lindbladian for a duration $s>0$, then apply the unitary control $e^{-isH_c}$, where
\begin{equation*}
 H_c=-iz\ket{v}\!\bra{\psi}+iz\ket{\psi}\!\bra{v},
\end{equation*}
with $z=(\mu_--\mu_+)/4$ and $\ket{v}=(\ket{00}-\ket{11})/\sqrt2$.
We then perform the projective measurement $\{P,\1-P\}$ with $P=\kb{00}+\kb{11}=\kb{\psi}+\kb{v}$, so that $\operatorname{im}(P)=\Span\{\ket{\psi},\ket{v}\}$.
We repeat this round $n$ times, with total exposure time $\tau=ns$, and finally measure $\{\kb{v},\1-\kb{v}\}$.
We accept the null hypothesis if every intermediate measurement returns $P$ and the final measurement returns $\kb{v}$.

To see why this protocol distinguishes the two hypotheses, consider the Zeno limit $n\to\infty$ with $s=\tau/n\to0$ at fixed total exposure time $\tau>0$.
On the branch where every projection succeeds, the limiting evolution under the alternative hypothesis is governed by the jumps and controlled drift projected onto $\operatorname{im}(P)$, as established in Section~\ref{sec:zeno-common} below.
Here both projected jumps vanish, and the control cancels the drift towards $\ket{v}$:
\begin{equation*}
 P(L_\pm\otimes\1_E)P=0,\qquad
 \left(P(G\otimes\1_E)P-iH_c\right)\ket{\psi}
 =-\frac{\mu_++\mu_-}{4}\ket{\psi}.
\end{equation*}
Thus, in the Zeno limit, the state under the alternative hypothesis remains $\kb{\psi}$ conditioned on every projection succeeding, so the final measurement never returns $\kb{v}$.
Under the identity null hypothesis, every projection succeeds and the control rotates $\ket{\psi}$ towards $\ket{v}$.
For finite $n$ with $s=\tau/n$, the protocol therefore has type-I and type-II error probabilities given by
\begin{equation*}
 \alpha=\cos^2(z\tau),\qquad \beta=O_\tau(s),
\end{equation*}
where the second estimate follows from the common product bound \eqref{sm:product} proved below.
In particular, $\tau=\pi/(2|z|)$ gives $\alpha=0$ and $\beta\to0$ as $s\to0$.
Theorem~\ref{thm:identity} determines the optimal threshold, which we evaluate for this family in Section~\ref{sec:generalised-amplitude-damping} together with its classification in Table~\ref{tab:regimes}.

\subsection{The common projected evolution}\label{sec:zeno-common}

We now establish the projected evolution and its approximation by finite protocols, which are used in the proof of Theorem~\ref{thm:infinite} below as well as in the achievability proof of Theorem~\ref{thm:identity} in Section~\ref{sec:identity}.
The discrimination mechanism is related to interaction-free channel discrimination~\cite{HasenoehrlWolf2022}, where repeated channel applications suppress a controlled transition under one hypothesis.
Here each unknown channel application has duration $t/n$, and the estimates below control the product as $n\to\infty$ at fixed total exposure time $t$.
Fixed-channel convergence estimates do not directly establish this limit, since the unknown channels themselves approach the identity as $n$ increases.

Related protocols for discriminating Lindbladians with identical jump operators and different Hamiltonians rely on quantum error correction methods combined with quantum Zeno dynamics to preserve a logical Hamiltonian difference~\cite{ChenYuan2019}.
In our protocol, the projection and Hamiltonian control are chosen so that, under the alternative hypothesis, the state remains $\kb{\psi}$ in the Zeno limit, conditioned on every projection succeeding.
In this limit, the final measurement cannot find the orthogonal target state $\ket{v}$ under this hypothesis, while a failed projection also leads to rejection of the null hypothesis.

In the preceding generalised amplitude damping example, all projected jumps under the alternative hypothesis vanish.
For a general alternative hypothesis, they may be nonzero, provided that $\ket{\psi}$ is an eigenvector of each projected jump.
Let $\ket{\psi},\ket{v}$ be orthogonal unit vectors on $AE$ such that $\bra{v}\left(L_j^{(2)}\otimes\1_E\right)\ket{\psi}=0$ for every $j$.
We define $P=\kb{\psi}+\kb{v}$, $z=\bra{v}(G_2\otimes\1_E)\ket{\psi}$ and
\begin{equation*}
 H_c=-iz\ket{v}\!\bra{\psi}+i\overline z\ket{\psi}\!\bra{v}.
\end{equation*}
As in the preceding example, we initialise the probe and auxiliary system in $\kb{\psi}$.
In each round, we let the probe evolve under the unknown Lindbladian dynamics for some time $s>0$.
We then apply the common unitary control $e^{-isH_c}$ and perform the projective measurement $\{P,\1-P\}$.
Since $H_c=PH_cP$, the control commutes with the projection.
We now fix the total exposure time $\tau>0$ and repeat this round $n$ times with $s=\tau/n$, so that the duration of each round tends to zero as $n\to\infty$.
Under hypothesis $i=1,2$, the branch in which every intermediate measurement returns $P$ has the subnormalised state
\begin{equation}
 \left[Pe^{-i\tau H_c/n}\left(e^{\tau\cL_i/n}\otimes\id_E\right)(\,\cdot\,)e^{i\tau H_c/n}P\right]^n\left(\kb{\psi}\right)
 \xrightarrow[n\to\infty]{}e^{\tau\cK_i}\left(\kb{\psi}\right).
 \label{sm:zeno-limit}
\end{equation}
Here we used the quantum Zeno limit for open systems~\cite{MobusWolf2019,Burgarth2020,BeckerDattaSalzmann2021}, combined with the Lie-Trotter product formula to incorporate the unitary control.
The resulting Zeno generator acts on operators $X=PXP$ as
\begin{equation}
 \cK_i(X)=P(\cL_i\otimes\id_E)(X)P-i[H_c,X].
 \label{sm:projected}
\end{equation}
The trace of the subnormalised state $e^{\tau\cK_i}\left(\kb{\psi}\right)$ is the probability that every projection succeeds in the Zeno limit.

If any intermediate measurement returns $\1-P$, we accept the alternative hypothesis.
If every intermediate measurement returns $P$, we finally measure $\{\kb{v},\1-\kb{v}\}$ and accept the null hypothesis for outcome $\kb{v}$ and the alternative hypothesis otherwise.
All event probabilities below are unconditional.

To quantify the errors of this decision rule, we first express them in terms of the successful branch and then estimate its distance from the Zeno evolution.
For $s=\tau/n$, the type-I and type-II error probabilities are
\begin{equation}
\begin{aligned}
 \alpha_n&=1-\bra{v}\left[Pe^{-isH_c}\left(e^{s\cL_1}\otimes\id_E\right)(\,\cdot\,)e^{isH_c}P\right]^n\left(\kb{\psi}\right)\ket{v}=1-\bra{v}e^{\tau\cK_1}\left(\kb{\psi}\right)\ket{v}+O(\tau^2/n),\\
 \beta_n&=\bra{v}\left[Pe^{-isH_c}\left(e^{s\cL_2}\otimes\id_E\right)(\,\cdot\,)e^{isH_c}P\right]^n\left(\kb{\psi}\right)\ket{v}=\bra{v}e^{\tau\cK_2}\left(\kb{\psi}\right)\ket{v}+O(\tau^2/n).
\end{aligned}
\label{sm:error-equalities}
\end{equation}
Here, we have estimated the Zeno and Trotter errors in the limit \eqref{sm:zeno-limit} at fixed total exposure time $\tau$ using a second-order Taylor expansion and telescoping.
In particular, for a Hermitian operator $X=PXP$ and $0\le s\le s_1$, the one-round error satisfies
\begin{equation*}
 \normone{Pe^{-isH_c}\left(e^{s\cL_i}\otimes\id_E\right)(X)e^{isH_c}P-e^{s\cK_i}(X)}\le C_0s^2\normone X,
\end{equation*}
with constants $C_0,s_1>0$ valid for both hypotheses.
Since each successful-round map and its Zeno limit are completely positive and trace non-increasing, they contract the trace norm of Hermitian operators, giving, for $ns\le\tau$,
\begin{equation}
 \normone{\left[Pe^{-isH_c}\left(e^{s\cL_i}\otimes\id_E\right)(\,\cdot\,)e^{isH_c}P\right]^n(X)-e^{ns\cK_i}(X)}\le C_0ns^2\normone X\le C_0\tau s\normone X.
 \label{sm:product}
\end{equation}

To estimate the type-II error probability $\beta_n$ further, we use that the controlled drift $P(G_2\otimes\1_E)P-iH_c$ and every projected jump operator $P(L_j^{(2)}\otimes\1_E)P$ map $\ket{\psi}$ to a scalar multiple of itself.
Thus, by trace preservation of $\cL_2$,
\begin{equation}
 e^{\tau\cK_2}\left(\kb{\psi}\right)=e^{-\lambda_2\tau}\kb{\psi},\qquad
 \lambda_2=\sum_j\left\|(\1-P)\left(L_j^{(2)}\otimes\1_E\right)\ket{\psi}\right\|^2\ge0.
 \label{sm:invariant-line}
\end{equation}
Since $\langle v|\psi\rangle=0$, the limiting subnormalised state $e^{-\lambda_2\tau}\kb{\psi}$ under the alternative hypothesis gives zero probability to the accepting outcome $\kb{v}$.
Inserting this into \eqref{sm:error-equalities}, with the approximation error bounded by \eqref{sm:product} at $s=\tau/n\le s_1$, gives
\begin{equation}
 \left|\alpha_n-\left[1-\bra{v}e^{\tau\cK_1}\left(\kb{\psi}\right)\ket{v}\right]\right|\le\frac{C_0\tau^2}{n},\qquad
 \beta_n\le\frac{C_0\tau^2}{n}.
 \label{sm:block-errors}
\end{equation}
This construction provides a single block of the protocols used below.
In the proof of Theorem~\ref{thm:infinite}, we choose $\ket{\psi},\ket{v}$ so that the transition probability $\bra{v}e^{\tau\cK_1}\left(\kb{\psi}\right)\ket{v}$ is positive, then repeat the block to achieve $\alpha\le\eps$.
In the proof of Theorem~\ref{thm:identity}, we calculate this probability explicitly and choose $\tau$ so that a single block already satisfies the type-I error constraint.

\subsection{General pairs: a transition and amplification}\label{sec:zeno-general}

\begin{proof}[Proof of Theorem~\ref{thm:infinite}]
We use the protocol of Section~\ref{sec:zeno-common} and repeat it independently to achieve the desired type-I error probability $\alpha\le\eps$.
More precisely, we first choose its input and target vectors, taking $E\simeq A$ and $\ket{\psi}=d^{-1/2}\sum_a\ket a_A\ket a_E$.
Choose $X$ to be a jump operator under the null hypothesis outside $\cS_2$, if one exists, and otherwise take $X=G_1-G_2\notin\cS_2$.
Since $X\mapsto(X\otimes\1_E)\ket{\psi}$ is injective, $(X\otimes\1_E)\ket{\psi}$ has a nonzero component orthogonal to $\Span\{\ket{\psi},(L_k^{(2)}\otimes\1_E)\ket{\psi}:k\}$.
Normalising this component gives $\ket{v}$ satisfying the orthogonality conditions of Section~\ref{sec:zeno-common}.

These vectors determine $P=\kb{\psi}+\kb{v}$ and the control $H_c$ from Section~\ref{sec:zeno-common}.
One run of this protocol starts with preparing $\kb{\psi}$ and ends with a decision between the hypotheses.
During each run, we perform up to $n$ rounds, each comprising evolution under the unknown dynamics for time $s=\tau/n$, the control $e^{-isH_c}$ and the measurement $\{P,\1-P\}$.
We accept the null hypothesis if all intermediate outcomes are $P$ and the final measurement returns $\kb{v}$, and accept the alternative hypothesis otherwise.
Its total exposure time is at most $\tau$.

To estimate the probability of accepting the null hypothesis in one run, define $p(u)=\bra{v}e^{u\cK_1}\left(\kb{\psi}\right)\ket{v}$ for $u\ge0$, with $\cK_1$ from \eqref{sm:projected}.
By \eqref{sm:block-errors}, $p(\tau)$ is the acceptance probability under the null hypothesis in the Zeno limit.
We first show that $p(\tau)>0$ for sufficiently small $\tau>0$ and control the errors for finite $n$.
We then repeat the protocol independently and combine the decisions from these runs to obtain \eqref{sm:achievability} for $t\ge t_\eps$, with
\begin{equation}
 p_0=\frac{p(\tau)}2,\qquad t_\eps=\tau\left(2+\frac8{p_0}\log\frac1\eps\right),\qquad c=\frac{p_0}{4\tau}.
 \label{sm:amplification-parameters}
\end{equation}

To establish this positivity, we use the Zeno generator \eqref{sm:projected}.
If $X$ is a jump operator under the null hypothesis, then $p(0)=0$ and
\begin{equation*}
 p'(0)=\sum_j\left|\bra{v}\left(L_j^{(1)}\otimes\1_E\right)\ket{\psi}\right|^2>0.
\end{equation*}
The sum is strictly positive because the term corresponding to $X$ is nonzero by our choice of $\ket{v}$.
Taylor expansion gives $p(u)=p'(0)u+O(u^2)$, so for sufficiently small $u>0$ the remainder has magnitude at most $p'(0)u/2$, yielding $p(u)\ge p'(0)u/2>0$.
Otherwise, $X=G_1-G_2$ and our choice of $\ket{v}$ gives
\begin{equation*}
 d_0=\bra{v}((G_1-G_2)\otimes\1_E)\ket{\psi}\ne0.
\end{equation*}
The common control subtracts the matrix element of $G_2$, so the no-jump contribution gives
\begin{equation*}
 p(u)\ge\left|\bra{v}e^{u(P(G_1\otimes\1_E)P-iH_c)}\ket{\psi}\right|^2
 =|d_0u+O(u^2)|^2.
\end{equation*}
For sufficiently small $u>0$, the remainder inside the modulus has magnitude at most $|d_0|u/2$, so $p(u)\ge |d_0|^2u^2/4>0$.
Thus, in either case, $p(u)>0$ for every sufficiently small positive $u$.
We fix one such value $\tau$, independently of the resolution time $\smin$.

To implement one run of the protocol in Section~\ref{sec:zeno-common} at resolution time $\smin$, we now choose the number $n$ of rounds and their duration $s$.
For $0<\smin\le\min\{\tau/2,s_1/2\}$, define
\begin{equation}
 n=\lfloor\tau/\smin\rfloor,\qquad s=\tau/n,\qquad \smin\le s\le2\smin.
 \label{sm:block-resolution}
\end{equation}
The $n$ rounds have total exposure time $\tau$.
For the type-I and type-II error probabilities $\alpha_n,\beta_n$ of one run, \eqref{sm:block-errors} gives
\begin{equation}
 \left|\alpha_n-(1-p(\tau))\right|\le2C_0\tau\smin,\qquad
 \beta_n\le2C_0\tau\smin.
 \label{sm:resolution-errors}
\end{equation}
Define $C=2C_0\tau$ and choose
\begin{equation}
 0<s_0\le\min\{\tau/2,s_1/2,p_0/(4eC)\}.
 \label{sm:cutoff}
\end{equation}
For $0<\smin\le s_0$, one run therefore has $\alpha_n\le1-p_0$ and $\beta_n\le C\smin$.

To reduce the type-I error below $\eps$, repeat the protocol $m=\lfloor t/\tau\rfloor$ times independently, preparing $\kb{\psi}$ again at the start of each run and recording its decision.
This uses total exposure time at most $m\tau\le t$.
Accept the null hypothesis if at least $k=\lceil p_0m/2\rceil$ runs accept it, and the alternative hypothesis otherwise.
For $t\ge t_\eps$, the resulting errors satisfy
\begin{align*}
 \alpha&\le2^{p_0m/2}(1-p_0/2)^m\le2^{-p_0m/8}\le\eps,\\
 \beta&\le\binom{m}{k}\beta_n^k
 \le\left(\frac{2eC\smin}{p_0}\right)^k
 \le\left(\frac{\smin}{2s_0}\right)^{p_0m/2}
 \le\left(\frac{\smin}{2s_0}\right)^{ct}.
\end{align*}
Here the first bound follows from Markov's inequality applied to $2^{-N}$, where $N$ counts runs that accept the null hypothesis and $\mathbb E[2^{-N}]\le(1-p_0/2)^m$ under the null hypothesis.
The second follows from a union bound over all choices of $k$ runs that accept the null hypothesis and \eqref{sm:cutoff}.
The final inequalities use $p_0m/8\ge\log(1/\eps)$ and $m\ge t/(2\tau)$, which follow from \eqref{sm:amplification-parameters}.
This proves \eqref{sm:achievability}, with $K=\tau(2+8/p_0)$ and $c,s_0,K$ depending only on the generators.
\end{proof}

\section{The uniform logarithmic converse}\label{sec:log}

We next prove that the logarithmic gain from Section~\ref{sec:zeno} is the largest possible when every positive-duration channel pair obeys support inclusion.
The key is to control regularised relative entropy using both analyticity and the common identity limit.

\begin{theorem}[Uniform logarithmic law for the resolution time]\label{thm:resolution}
Suppose that, for all $s>0$,
\begin{equation}
 \supp J^{(1)}_{AA'}(s)\subseteq\supp J^{(2)}_{AA'}(s).
 \label{sm:support}
\end{equation}
There are constants $C,s_*>0$, depending only on the generators, such that, for all $0<s\le s_*$,
\begin{equation}
 \frac{\DA\left(e^{s\cL_1}\middle\|e^{s\cL_2}\right)}s\le C[1+\log(s_*/s)].
 \label{sm:localupper}
\end{equation}
Moreover, if $\gamma(\cL_1\|\cL_2)=\infty$, the reverse inequality holds with another positive constant for all sufficiently small $s$.
Furthermore, in this case of infinite Stein exponent, for all fixed $0<\eps<1$, there are $s_0>0$, $0<c_\eps\le C_\eps$ and $t_\eps<\infty$, depending only on the generators and $\eps$, such that
\begin{equation}
 \left(\frac{\smin}{2s_0}\right)^{C_\eps t}
 \le\bmin(\cL_1\|\cL_2)
 \le\left(\frac{\smin}{2s_0}\right)^{c_\eps t},
 \label{sm:jointlaw}
\end{equation}
where $t\ge t_\eps$ and $0<\smin\le s_0$.
Moreover, the threshold $t_\eps$ can be chosen to satisfy $t_\eps\le K[1+\log(1/\eps)]$, where $K>0$ depends only on the generators.
Furthermore, for all $\smin>0$ and $t\ge0$, $\bmin(\cL_1\|\cL_2)>0$.
In particular, when $\gamma(\cL_1\|\cL_2)=\infty$, the upper bound in \eqref{sm:jointlaw} gives $\bet(\cL_1\|\cL_2)=\lim_{\smin\to0}\bmin(\cL_1\|\cL_2)=0$ for all $t\ge t_\eps$.
\end{theorem}

The upper bound on the optimal type-II error probability in \eqref{sm:jointlaw} was proved in Section~\ref{sec:zeno-general}.
Our main task here is to prove the matching lower bound for arbitrary adaptive protocols.
We first consider $n$ parallel channel uses of duration $s$ and write the dilation output under the null hypothesis as a polynomial in an auxiliary parameter.
We approximate this polynomial by keeping only terms up to a suitable degree, and use the short-time behaviour of the channels to show that the required degree depends only on the total exposure time $ns$.
Analyticity of the Choi states then lets us obtain this approximation from the output under the alternative hypothesis through an environment map whose norm we can bound.
This comparison bounds the type-II error from below whenever the test accepts the null hypothesis with probability at least $1/2$.
The state Stein lemma converts this testing bound into the relative entropy estimate \eqref{sm:localupper}, and the entropy bound for adaptive protocols \eqref{sm:telescopebound} gives the desired lower bound in \eqref{sm:jointlaw}.
Finally, we use the protocol of Section~\ref{sec:zeno-general} to obtain the matching lower bound on relative entropy per unit time.

We first prove a lemma showing that the Choi rank is constant at positive times, which we use to bound the smallest positive Choi eigenvalue, and then give the proof of Theorem~\ref{thm:resolution}.

\begin{lemma}\label{lem:rank}
Let $\cL$ be  finite-dimensional Lindbladian. Then we have that the rank of $J_{AA'}\left(e^{s\cL}\right)$ is constant on $s>0$.
\end{lemma}

\begin{proof}
To compare the Choi ranks at different times, we retain the no-jump contribution during the additional evolution interval.
Write $\cL(X)=GX+XG^*+\cJ(X)$, where $\cJ(X)=\sum_jL_jXL_j^*$, and define $\mathcal F_u(X)=e^{uG}Xe^{uG^*}$ for $u\ge0$.
Duhamel's formula gives
\begin{equation*}
 e^{u\cL}=\mathcal F_u+\int_0^u\mathcal F_{u-v}\circ\cJ\circ e^{v\cL}\,\dd v.
\end{equation*}
Every map in the integrand is completely positive, so $e^{u\cL}-\mathcal F_u$ is completely positive.
Consequently, applying $(e^{u\cL}-\mathcal F_u)\otimes\id_{A'}$ to the positive operator $J_{AA'}\left(e^{s\cL}\right)$ gives a positive operator.
Using the semigroup property and the definition of the Choi state, we therefore obtain, for $s,u>0$,
\begin{align*}
 J_{AA'}\left(e^{(s+u)\cL}\right)
 &=\left(e^{u\cL}\otimes\id_{A'}\right)\left(J_{AA'}\left(e^{s\cL}\right)\right)\\
 &\ge\left(\mathcal F_u\otimes\id_{A'}\right)\left(J_{AA'}\left(e^{s\cL}\right)\right)\\
 &=(e^{uG}\otimes\1_{A'})J_{AA'}\left(e^{s\cL}\right)(e^{uG^*}\otimes\1_{A'}).
\end{align*}
Since $e^{uG}$ is invertible, the preceding inequality shows that the rank of the Choi state $J_{AA'}\left(e^{s\cL}\right)$ cannot decrease with $s$.
Suppose it increased from $r$ at some time $s_1>0$ to a larger value at $s_2>s_1$.
Then some minor of order $r+1$, meaning the determinant of an $(r+1)\times(r+1)$ submatrix, would be nonzero at $s_2$.
However, the rank is at most $r$ for $0<s<s_1$, so this minor vanishes throughout that interval.
This is impossible because each minor is real analytic in $s$, and an analytic function that vanishes on an open interval must vanish identically.
\end{proof}

\begin{proof}[Proof of Theorem~\ref{thm:resolution}]
To prove the lower bound on the optimal type-II error probability in \eqref{sm:jointlaw}, we establish the relative entropy estimate \eqref{sm:localupper}.
For one channel use, support inclusion ensures that an outcome possible under the null hypothesis is also possible under the alternative hypothesis.
To quantify how much smaller its probability can be under the alternative hypothesis, we first bound the smallest positive Choi eigenvalue under that hypothesis from below.
Let $J^{(i)}_{AA'}(s)$ denote the normalised Choi state of the channel under hypothesis $i$ after exposure time $s>0$.
Let $r$ denote the rank of $J^{(2)}_{AA'}(s)$, which is constant for positive $s$ by Lemma~\ref{lem:rank}.
Consider the product $P(s)=\prod_{j=1}^r\lambda_j(s)$, where $\lambda_1(s),\ldots,\lambda_r(s)$ are its positive eigenvalues.
This product equals the sum of all principal minors of order $r$, meaning determinants of submatrices obtained by selecting the same $r$ row and column indices, so it extends analytically to $s=0$, where the rank may drop.
Since $P(s)>0$ for $s>0$, its Taylor expansion at zero has a first nonzero term.
Every eigenvalue is at most one, so the smallest positive eigenvalue $\lambda_{\min}^+$ is at least $P(s)$.
Consequently, there are $s_*,a_*>0$ and an integer $q\ge1$ such that, for $0<s\le s_*$,
\begin{equation}
 \lambda_{\min}^+\left(J^{(2)}_{AA'}(s)\right)\ge P(s)\ge a_*\ \left(\frac{s}{s_*}\right)^q.
 \label{sm:choieigen}
\end{equation}

Define $\cR_s=e^{s\cL_1}-e^{sG_1}(\,\cdot\,)e^{sG_1^*}$, the remainder after subtracting the no-jump contribution, which is completely positive by Duhamel's formula as in the proof of Lemma~\ref{lem:rank}.
Let $d_G=\op{G_1-G_2}$ and
$\kappa=\op{\sum_jL_j^{(1)*}L_j^{(1)}}$.
The operators $e^{sG_i}$ are contractions in operator norm.\footnote{Since $G_i+G_i^*=-\sum_jL_j^{(i)*}L_j^{(i)}$, for $s\ge0$ and any vector $\ket{\phi}$ we have $\frac{\dd}{\dd s}\|e^{sG_i}\ket{\phi}\|^2=-\sum_j\|L_j^{(i)}e^{sG_i}\ket{\phi}\|^2\le0$.
Thus $\|e^{sG_i}\ket{\phi}\|\le\|\ket{\phi}\|$, giving $\op{e^{sG_i}}\le1$.}
Duhamel's formula and trace preservation therefore give
\begin{equation}
 \op{e^{sG_1}-e^{sG_2}}\le d_Gs,\qquad
 \cR_s^*(\1)=\1-e^{sG_1^*}e^{sG_1}\le\kappa s\1.
 \label{sm:nearidentity}
\end{equation}
Choose Kraus operators $M_{s,j}$ for $\cR_s$ and a common finite dilation space $F$ for the two channels, padding the dilations if necessary. Define
\begin{equation*}
 W_s=\left(e^{sG_1}-e^{sG_2}\right)\otimes\ket0_F+\sum_jM_{s,j}\otimes\ket j_F,\qquad
 V_s(z)=e^{sG_2}\otimes\ket0_F+zW_s .
\end{equation*}
The labels $j$ in this expression are orthogonal to $0$.
At $z=1$, $V_s(1)$ is a Stinespring isometry for $e^{s\cL_1}$.
For $|z|=2$, \eqref{sm:nearidentity} yields
\begin{equation*}
 V_s(z)^*V_s(z)
 \le[1+(4d_G+4d_G^2s_*+4\kappa)s]\1 .
\end{equation*}
Thus, for a fixed $a_0\ge0$, we have $\op{V_s(z)}\le2^{a_0s}$ on this circle.

To compare this dilation with the channel under the alternative hypothesis, consider the completely positive map $\cB_s=\left(e^{sG_1}-e^{sG_2}\right)(\,\cdot\,)\left(e^{sG_1^*}-e^{sG_2^*}\right)+\cR_s$.
Its Choi support lies inside $\supp J^{(2)}_{AA'}(s)$.
For $\cR_s$ and the Kraus vector of $e^{sG_1}$ this follows from \eqref{sm:support}. The Kraus vector of $e^{sG_2}$ already belongs to that support, so the same holds for their difference.
Its Choi trace is at most $d_G^2s^2+\kappa s$.
\eqref{sm:choieigen} therefore gives
\begin{equation}
 \cB_s\cpleq b_s^2e^{s\cL_2},\qquad
 0\le b_s\le b_*(s_*/s)^{(q-1)/2}
 \label{sm:bs}
\end{equation}
for some finite $b_*$.

Choose a dilation $V_2(s)$ of the channel under the alternative hypothesis with $e^{sG_2}$ as its explicit no-jump Kraus operator.
There are environment maps $A_s,B_s$ with
\begin{equation*}
 (\1_A\otimes A_s)V_2(s)=e^{sG_2}\otimes\ket0_F,\quad
 (\1_A\otimes B_s)V_2(s)=W_s,\quad
 \op{A_s}\le1,\quad \op{B_s}\le b_s.
\end{equation*}
For completeness, the second statement is the finite-dimensional factorisation behind completely positive order.
If the Choi factors for $\cB_s$ and $e^{s\cL_2}$ are $XX^*$ and $YY^*$, their order gives $X=YC$ with $\op C\le b_s$ by the pseudoinverse on the range of $Y$.
Transposing the coefficient matrix gives the indicated map on Kraus labels.
The first map simply selects the explicit no-jump Kraus operator.

To control the approximation uniformly over inputs, we truncate the joint dilation polynomial at a degree depending only on the total exposure time $ns$.
For any pure input $\ket{\psi}_{A^nE}$ on $A^nE$, identify the joint dilation output $(AF)^nE$ with $A^nEF^n$ by reordering tensor factors, and consider
\begin{align*}
 \ket{\Psi(z)}_{A^nEF^n}&=\left(V_s(z)^{\otimes n}\otimes\1_E\right)\ket{\psi}_{A^nE}
       =\sum_{j=0}^nz^j\ket{\Psi_j}_{A^nEF^n},\qquad
 \ket{\Psi^{(2)}}_{A^nEF^n}=\left(V_2(s)^{\otimes n}\otimes\1_E\right)\ket{\psi}_{A^nE} .
\end{align*}
On $|z|=2$ the norm is at most $2^{a_0ns}$.
Cauchy's coefficient estimate gives
$\|\ket{\Psi_j}_{A^nEF^n}\|\le2^{a_0ns-j}$.
Choose $K=\lceil a_0ns+2\rceil$ and truncate the polynomial at degree $K$, or retain it all if $K\ge n$.
The omitted norm at $z=1$ is at most $2^{a_0ns-K}\le1/4$.
The retained vector equals $(\1_{A^nE}\otimes C_{\le K})\ket{\Psi^{(2)}}_{A^nEF^n}$, where
\begin{equation*}
 \op{C_{\le K}}
 \le\sum_{j=0}^{\min(K,n)}\binom{n}{j}b_s^j
 \le e^{ns/s_*}\left[\frac{s_*}{s}(1+b_s)\right]^K .
\end{equation*}
To prove the last inequality, set $x=s/[s_*(1+b_s)]\le1$.
For $j\le K$, $x^j\ge x^K$, and hence the sum is bounded by
$x^{-K}(1+xb_s)^n\le x^{-K}e^{ns/s_*}$.
This estimate avoids a spurious factor $\log n$.

Consider any final binary measurement $\{Q,\1-Q\}$ on the accessible output $A^nE$, where $0\le Q\le\1$, which accepts the null hypothesis for outcome $Q$ and the alternative hypothesis otherwise.
Suppose the probability of outcome $Q$ under the null hypothesis is at least $1/2$.
Its acceptance amplitude on the retained vector is at least
$1/\sqrt2-1/4\ge1/(2\sqrt2)$.
Since $Q$ commutes with the environment map, this amplitude is at most
$\op{C_{\le K}}\sqrt\beta$, where $\beta$ is the probability of outcome $Q$ under the alternative hypothesis, hence the type-II error probability.
It follows that
\begin{equation}
 -\log\beta\le
 3+2(\log e)\frac{ns}{s_*}
 +2\lceil a_0ns+2\rceil
       \log\left[\frac{s_*}{s}(1+b_s)\right].
 \label{sm:polynomialtest}
\end{equation}
All constants are independent of the number of probes and the auxiliary system.

To obtain the regularised relative entropy bound \eqref{sm:localupper}, fix $s$, a block size $k$ and an input in \eqref{sm:chain}.
Repeat the block $r$ times and apply the state Stein lemma with acceptance at least $1/2$.
\eqref{sm:polynomialtest} holds for $n=rk$.
Divide by $n$ and let $r\to\infty$.
Every achievable state error exponent is bounded, so
\begin{align*}
 \frac1kD\left(\left(\left(e^{s\cL_1}\right)^{\otimes k}\otimes\id_E\right)(\omega_{A^kE})\|
                \left(\left(e^{s\cL_2}\right)^{\otimes k}\otimes\id_E\right)(\omega_{A^kE})\right)
 &\le 2(\log e)\frac{s}{s_*}
 +2a_0s\log\left[\frac{s_*}{s}(1+b_s)\right].
\end{align*}
Taking the suprema over $k$ and $\omega_{A^kE}$, using \eqref{sm:chain}, and then \eqref{sm:bs} proves \eqref{sm:localupper}.
The Stein limit was taken at each fixed $s$. No exchange with a simultaneous small-$s$ limit is required.

By \eqref{sm:subdivision}, for $\smin\le s_*/2$ this bounds
$\Gamma_{\smin}(\cL_1\|\cL_2)$ by a constant times $1+\log(s_*/\smin)$.
For all $\smin>0$, \eqref{sm:telescopebound} therefore gives positivity for all $t\ge0$, using monotonicity in $\smin$ if $\smin>s_*/2$.

If $\gamma(\cL_1\|\cL_2)=\infty$, choose a common smaller cutoff $s_0$ for this estimate and Theorem~\ref{thm:infinite}.
After adjusting a constant $B>0$ depending only on the generators, \eqref{sm:telescopebound} gives
\begin{equation*}
 -\log\bmin(\cL_1\|\cL_2)
 \le\frac{Bt[1+\log(s_0/\smin)]+h(\eps)}{1-\eps}.
\end{equation*}
Keep the choice of $t_\eps$ in \eqref{sm:amplification-parameters}, as used in the proof of Theorem~\ref{thm:infinite}.
Since $h(\eps)\le1$ and $t\ge t_\eps\ge2\tau$, the right-hand side is at most $C_\eps t[1+\log(s_0/\smin)]$ with $C_\eps=(B+1/(2\tau))/(1-\eps)$.
Exponentiating gives the lower bound in \eqref{sm:jointlaw}, while Theorem~\ref{thm:infinite} gives the upper bound.
Thus the same threshold satisfies the full estimate and its stated logarithmic upper bound.

To obtain the matching logarithmic lower bound on the relative entropy per unit time in this case, use the vectors and $\tau$ from Section~\ref{sec:zeno-general} and take $n=\lfloor\tau/s\rfloor$.
Since $0\le\tau-ns<s$, \eqref{sm:product} and continuous differentiability of $p(u)$ imply that the accepting probabilities satisfy $p\ge p_0=p(\tau)/2$ and $q\le C_1s$ for some $C_1>0$ and sufficiently small $s>0$.
The entropy budget and binary data processing then give, for $C_1s\le p_0/e$,
\begin{equation}
 \frac1s\DA\left(e^{s\cL_1}\middle\|e^{s\cL_2}\right)
 \ge\frac{D((p,1-p)\|(q,1-q))}{ns}
 \ge\frac{p_0}{\tau}\log\frac{p_0}{eC_1s}.
 \label{sm:event-entropy}
\end{equation}
Here we used $ns\le\tau$ and $D((p,1-p)\|(q,1-q))\ge p\log[p/(eq)]$.
Thus \eqref{sm:event-entropy} supplies the claimed lower bound for all sufficiently small $s$.

Finally, letting $\smin\to0$ in the upper bound in \eqref{sm:jointlaw} shows that the unrestricted optimal type-II error probability vanishes for all $t\ge t_\eps$, since it is no larger than the optimum at any positive resolution time.
\end{proof}

\section{The three regimes and attainment at positive resolution time}\label{sec:regimes}

We combine the results on the Stein exponent and resolution time into the classification in Table~\ref{tab:regimes}, illustrated in Figure~\ref{fig:regimes}.
For that, we first distinguish attainment of zero type-II error at positive resolution time from an optimal type-II error probability that tends to zero only as the resolution time decreases.

\begin{lemma}\label{lem:support-attainment}
For two finite-dimensional Lindbladians and a adaptive protocols $\mathsf P$ as defined in \eqref{eq:beta}, the following are equivalent:
\begin{enumerate}[label=(\roman*)]
 \item\label{item:support-violation} $\supp J^{(1)}_{AA'}(s)\not\subseteq\supp J^{(2)}_{AA'}(s)$ for some $s>0$,
 \item\label{item:zero-error-attainment} for all $\smin>0$ and $0<\eps<1$, there exist a finite $t\ge0$ and a protocol $\mathsf P$ with total exposure time at most $t$ and resolution time $\smin$ such that $\alpha(\mathsf P)\le\eps$ and $\beta(\mathsf P)=0$.
\end{enumerate}
If instead support inclusion holds at all positive times, then for all protocols $\mathsf P$  with total time $t\ge0$ and time $\smin>0$ fixed and further $\alpha(\mathsf P)<1$ has $\beta(\mathsf P)>0$.
\end{lemma}

\begin{proof}
We first prove \ref{item:support-violation}$\Rightarrow$\ref{item:zero-error-attainment}.
By assumption, support inclusion fails at some time $s>0$.
We show that, for any prescribed resolution time $\smin>0$, it also fails at some time $\tilde s\ge\smin$.
Write $J^{(i)}_{AA'}(u)=J_{AA'}\left(e^{u\cL_i}\right)$ and let $r_2$ be the rank of $J^{(2)}_{AA'}(u)$, which is constant for $u>0$ by Lemma~\ref{lem:rank}.
Since the Choi states are positive, for $u>0$ we have
\begin{equation*}
 \supp J^{(1)}_{AA'}(u)\subseteq\supp J^{(2)}_{AA'}(u)
 \quad\Longleftrightarrow\quad
 \operatorname{rank}\left(J^{(1)}_{AA'}(u)+J^{(2)}_{AA'}(u)\right)=r_2.
\end{equation*}
Since support inclusion fails at $u=s$ by \ref{item:support-violation}, we have $\operatorname{rank}\left(J^{(1)}_{AA'}(s)+J^{(2)}_{AA'}(s)\right)>r_2$.
Hence some minor of order $r_2+1$ of the sum $J^{(1)}_{AA'}(u)+J^{(2)}_{AA'}(u)$ is nonzero at $u=s$.
This minor is real analytic in $u$ and not identically zero, so it cannot vanish throughout $(\smin,2\smin)$.
Hence support inclusion also fails at some $\tilde s\in(\smin,2\smin)$, providing an allowed evolution duration.

To construct the protocol $\mathsf P$, let $Q$ be the projection onto $\ker J^{(2)}_{AA'}(\tilde s)$.
In each run, prepare the maximally entangled state $\ket{\Phi}=d^{-1/2}\sum_{a=1}^d\ket a_A\ket a_{A'}$, with $d=\dim A$ and $A'\simeq A$, and let $A$ evolve under the unknown dynamics for time $\tilde s$, leaving $A'$ unchanged.
The resulting state under hypothesis $i$ is $J^{(i)}_{AA'}(\tilde s)$, which we measure with $\{Q,\1-Q\}$.
By support violation, outcome $Q$ has probability $p=\Tr\left(QJ^{(1)}_{AA'}(\tilde s)\right)>0$ under the null hypothesis and zero under the alternative hypothesis.
Choose an integer $m\ge1$ such that $(1-p)^m\le\eps$, and repeat this preparation, evolution and measurement independently $m$ times.
The protocol accepts the null hypothesis if outcome $Q$ occurs in at least one run, and the alternative hypothesis otherwise.
Thus $\alpha(\mathsf P)=(1-p)^m\le\eps$ and $\beta(\mathsf P)=0$, with total exposure time $t=m\tilde s$ and evolution durations $\tilde s\ge\smin$ as required in \eqref{eq:beta}.

We now prove \ref{item:zero-error-attainment}$\Rightarrow$\ref{item:support-violation} by contraposition.
Suppose support inclusion holds at all positive times and fix a protocol $\mathsf P$ with total exposure time at most $t<\infty$, resolution time $\smin>0$ and $\alpha(\mathsf P)<1$.
Choose $0<\eps<1$ with $\alpha(\mathsf P)\le\eps$.
By \eqref{eq:beta} and Theorem~\ref{thm:resolution}, $\beta(\mathsf P)\ge\bmin(\cL_1\|\cL_2)>0$.
This also proves the converse implication.
\end{proof}

\begin{corollary}\label{cor:trichotomy}
Exactly one of the following three cases occurs:
\begin{enumerate}[label=(\roman*),beginpenalty=10000]
 \item $\gamma(\cL_1\|\cL_2)<\infty$.
 The algebraic conditions in \eqref{sm:span} hold, and $\bet(\cL_1\|\cL_2)>0$ for all $0\le\eps<1$ and $t\ge0$.
 Moreover, for distinct generators and fixed $0<\eps<1$, $-\log\bet(\cL_1\|\cL_2)=\Theta(t)$.
 \item Support inclusion fails at some positive duration.
 Then $\gamma(\cL_1\|\cL_2)=\infty$, and for all $0<\eps<1$ and $\smin>0$, a finite protocol at that resolution time attains $\beta=0$ at sufficiently large finite total time.
 \item $\gamma(\cL_1\|\cL_2)=\infty$ and support inclusion holds at all positive durations.
 Then $\bmin(\cL_1\|\cL_2)>0$ for all $t\ge0$ and $\smin>0$, while $\bet(\cL_1\|\cL_2)=0$ at sufficiently large $t$ for all fixed $0<\eps<1$.
 More precisely, in the uniform range of \eqref{sm:jointlaw}, $-\log\bmin(\cL_1\|\cL_2)=\Theta(t[1+\log(s_0/\smin)])$.
\end{enumerate}
\end{corollary}

\begin{proof}
Finite $\gamma(\cL_1\|\cL_2)$ excludes support failure, because the latter gives an infinite channel entropy increment at that duration and hence an infinite Stein exponent by Theorem~\ref{thm:rate}.
The first case follows from Theorems~\ref{thm:rate} and \ref{thm:span}.
If $\gamma(\cL_1\|\cL_2)=\infty$, support inclusion either fails at some duration or holds at all of them.
Lemma~\ref{lem:support-attainment} gives the second case, and Theorems~\ref{thm:infinite} and \ref{thm:resolution} give the third.
\end{proof}

This proves the three-way refinement in Table~\ref{tab:regimes}.

\begin{proof}[Proof of Theorem~\ref{thm:letter-span}]
Theorem~\ref{thm:span} gives the equivalence of the first two conditions.
If these hold, Theorem~\ref{thm:rate} gives \eqref{eq:finitefloor}, which implies the third condition.
Conversely, suppose the third condition holds for some $0<\eps_0<1$.
If the span conditions failed, Theorem~\ref{thm:infinite} would give a finite $t_{\eps_0}$ such that $\beta_{\eps_0,t}(\cL_1\|\cL_2)=0$ for all $t\ge t_{\eps_0}$, contradicting this positivity.
\end{proof}

For comparison, consider a fixed pair of states $\rho,\sigma$ in finite dimension.
Support inclusion is equivalent to $\rho\le c\sigma$ for some finite $c$, and implies $D(\rho\|\sigma)<\infty$.
Without inclusion, the projection onto $\ker\sigma$ defines an event that is strictly forbidden under the alternative hypothesis but has positive probability under the null hypothesis, and the same repetition argument gives zero type-II error with any positive type-I tolerance.
For fixed channels $\cN,\cM$, \eqref{sm:cporder} gives the corresponding dichotomy.
If $\cN\cpleq c\cM$, then operator monotonicity and data processing give $\DA(\cN\|\cM)\le\log c$.
If support inclusion fails, Choi probes give an exactly forbidden event.
The difference for generators is that the duration is also optimised: the finite channel divergences in case (iii), divided by their durations, become unbounded as the duration tends to zero.
All these statements concern asymmetric testing at a positive type-I tolerance, rather than exact vanishing of both errors.

\section{The exact threshold for an identity null hypothesis}\label{sec:identity}

When the null hypothesis is the identity semigroup, we refine the common construction of Section~\ref{sec:zeno-common} by optimising its exact rotation under the null hypothesis.
Under the alternative hypothesis, the initial state remains unchanged in the Zeno limit, conditioned on obtaining outcome $P=\kb{\psi}+\kb{v}$ in every interleaved measurement $\{P,\1-P\}$ introduced in Section~\ref{sec:zeno-common}, even when the projected jump operators do not vanish.
A single run of the protocol described in Section~\ref{sec:zeno-common}, with a suitable total exposure time, then meets the type-I constraint without the statistical amplification used for general Lindbladian pairs in the proof of Theorem~\ref{thm:infinite}.
Let the null hypothesis be $\cL_1=0$, and write the generator under the alternative hypothesis as
$\cL_2(X)\equiv\cL(X)=GX+XG^*+\sum_jL_jXL_j^*$.
Define
\begin{equation}
 \cS=\Span\{\1,L_j:j\},\qquad
 g=\inf_{C\in\cS}\op{G-C}.
 \label{sm:g}
\end{equation}
Theorem~\ref{thm:span} gives $g=0$ if and only if $\gamma(0\|\cL)<\infty$.

\begin{theorem}[Identity null hypothesis]\label{thm:identity}
If $g>0$, then for all $0\le\eps<1$,
\begin{equation}
 \beta_{\eps,t}(0\|\cL)=0
 \quad\Longleftrightarrow\quad
 t\ge T_\eps:=\frac{\arccos\sqrt\eps}{g}.
 \label{sm:threshold}
\end{equation}
Moreover, at total time $T_\eps$, protocols with positive resolution time have exactly $\alpha=\eps$ and $\beta=O(\smin)$.
\end{theorem}

\begin{proof}[Proof of Theorem~\ref{thm:identity}]
We use the protocol of Section~\ref{sec:zeno-common} with an auxiliary system $E\simeq A$ and input and target states $\ket{\psi},\ket{v}$ on $AE$ chosen so that $|\bra{v}(G\otimes\1_E)\ket{\psi}|=g$.
We first show that such states satisfying the protocol's orthogonality conditions exist, then choose the total exposure time to attain the threshold and finally prove the converse below it.

For $g>0$, operator-norm and trace-norm duality on the quotient by $\cS$ gives
\begin{equation*}
 g=\max_{\substack{\normone W\le1\\
                   \Tr(W^*C)=0\ \forall C\in\cS}}
          |\Tr(W^*G)|.
\end{equation*}
Let $W$ be a maximiser, which has trace norm one because $g>0$.
Since $\1\in\cS$, the constraint also gives $\Tr W=0$.
Write its polar decomposition as $W=U|W|$.
Choose corresponding orthonormal bases $\{\ket a_A\}_a$ and $\{\ket a_E\}_a$ for $A$ and $E$.
For any operator $X$ on $A$, define its unnormalised vectorisation on $AE$ by $\ket{X}\!\rangle=(X\otimes\1_E)\sum_a\ket a_A\ket a_E$.
Applying this construction to $|W|^{1/2}$ and $U|W|^{1/2}$ gives the vectors
\begin{equation*}
 \ket{\psi}=\bigl||W|^{1/2}\bigr\rangle\!\rangle,\qquad
 \ket{v}=\bigl|U|W|^{1/2}\bigr\rangle\!\rangle
\end{equation*}
which are normalised and satisfy, for every operator $X$ on $A$,
$\bra{v}(X\otimes\1_E)\ket{\psi}=\Tr(W^*X)$.
Hence $\langle v|\psi\rangle=\Tr(W^*)=0$ and $\bra{v}(L_j\otimes\1_E)\ket{\psi}=0$ for all $j$.
Since $W$ is a maximiser, these vectors also satisfy
\begin{equation}
 |\bra{v}(G\otimes\1_E)\ket{\psi}|=|\Tr(W^*G)|=g.
 \label{sm:witnessg}
\end{equation}

To exploit the identity null hypothesis, apply the protocol of Section~\ref{sec:zeno-common} with $\cL_1=0$, $\cL_2=\cL$, the vectors constructed above and the control $H_c$ defined there.
The interleaved projection in the Zeno product \eqref{sm:zeno-limit} is $P=\kb{\psi}+\kb{v}$, and by \eqref{sm:witnessg} we can fix the phase of $\ket{v}$ so that $z=\bra{v}(G\otimes\1_E)\ket{\psi}=g>0$.
For $s=\tau/n\le s_1$, we have
\begin{equation}
 \alpha_n=\cos^2(|z|\tau),\qquad \beta_n\le\frac{C_0\tau^2}{n},
 \label{sm:identity-block}
\end{equation}
where the bound on $\beta_n$ follows from \eqref{sm:block-errors}, while the equality for $\alpha_n$ follows by using $\cL_1=0$ in \eqref{sm:error-equalities} and $H_c=PH_cP$, $H_c^2=|z|^2P$, so that every intermediate measurement returns $P$ and $\alpha_n=1-|\bra{v}e^{-i\tau H_c}\ket{\psi}|^2=1-\sin^2(|z|\tau)=\cos^2(|z|\tau)$.

For $0\le\eps<1$ and $0<\smin\le\min\{T_\eps/2,s_1/2\}$, choosing total exposure time $\tau=T_\eps=\arccos\sqrt\eps/g$ and the subdivision in \eqref{sm:block-resolution}, we obtain from \eqref{sm:identity-block}
\begin{equation*}
 \alpha_n=\cos^2(gT_\eps)=\eps,\qquad \beta_n\le2C_0T_\eps\smin.
\end{equation*}
Letting $\smin\to0$ proves that the optimal type-II error probability vanishes at the endpoint of \eqref{sm:threshold}.
Using the same protocol within any larger total exposure time budget proves the same for $t\ge T_\eps$.

To prove the converse for arbitrary adaptive protocols, choose a minimiser $C=a\1+\sum_jc_jL_j$ in \eqref{sm:g}, which in particular gives $\op{G-C}=g$.
To control the change of direction of the no-jump state by $g$, we choose an equivalent Lindblad representation in which the drift equals $G-C$ up to a scalar multiple of the identity.
Following \eqref{sm:gauge}, shift the jump operators to $L'_j=L_j+\overline c_j\1$ and replace the drift by
\begin{equation*}
 G'=G-\sum_jc_jL_j-\tfrac12\sum_j|c_j|^2\1
    =G-C+\zeta\1,\qquad
 \zeta=a-\tfrac12\sum_j|c_j|^2 .
\end{equation*}
This change of jump operators and drift leaves the Lindbladian $\cL$ invariant.
Let
$K=\sum_j(L'_j)^*L'_j=-(G'+G'^*)\ge0$ and $\nu=\op K$.
The channel under the alternative hypothesis contains its no-jump completely positive part:
\begin{equation}
 \cQ_s(X)=e^{sG'}Xe^{sG'^*},\qquad \cQ_s\cpleq e^{s\cL}.
 \label{sm:nojump}
\end{equation}

Consider an arbitrary finite adaptive protocol.
Purify the initial state and every common control, retaining a purification environment.
For an instrument, keep an orthogonal outcome record and an additional untouched copy so that discarding the latter implements the original classical record.
All subsequent operations are controlled by that record.
The original final test acts trivially on the added purification systems.

Pad any branch which stops at accumulated time $T_h<t$ with a call of duration $t-T_h$ on a fresh disposable probe.
Discarding that probe leaves the original output and errors unchanged.
The padding is only a converse device and need not respect a positive resolution time.
Now every branch has total time exactly $t$.
Synchronise the finitely many branch endpoints on the common axis for the accumulated time $u\in[0,t]$.
Between endpoints, every branch is being sampled, either on its physical probe or on a padding probe.
The no-jump generator under the alternative hypothesis is therefore a direct sum of copies of $G-C+\zeta\1$.
Its scalar part is a global scalar and its residual operator has norm at most $g$.

Take a common finite enlarged system $AEF$ for the protocol, with $F$ containing the retained purification systems and $E$ also containing the padding probes.
On $AEF$, let $\ket{\phi_u}$ be the normalised purified vector under the null hypothesis and let $\ket{\psi_u}$ be the subnormalised vector under the alternative hypothesis obtained by selecting the no-jump term in \eqref{sm:nojump} at every call.
Positivity of the omitted Dyson terms means that the output on $AEF$ under the alternative hypothesis dominates this no-jump substate.
Common isometries preserve its norm.
During each call,
\begin{equation*}
 \frac{\dd}{\dd u}\|\ket{\psi_u}\|^2
 =-\bra{\psi_u}K_u\ket{\psi_u}\ge-\nu\|\ket{\psi_u}\|^2,
\end{equation*}
where $K_u$ is the controlled copy of $K$.
Consequently,
\begin{equation}
 \|\ket{\psi_t}\|^2\ge e^{-\nu t}.
 \label{sm:survival}
\end{equation}

Normalise $\ket{\widehat{\psi}_u}=\ket{\psi_u}/\|\ket{\psi_u}\|$.
After removing the common isometries and the global scalar, its derivative orthogonal to itself is
$\left(\1-\kb{\widehat{\psi}_u}\right)(G-C)_u\ket{\widehat{\psi}_u}$, where $(G-C)_u$ is the controlled copy of $G-C$ on the enlarged system and has operator norm at most $g$.
Thus the projective angle traversed per unit total time is at most $g$.
Under the null hypothesis there is no unknown evolution, and common isometries preserve the angle.
Since the two vectors agree initially,
\begin{equation}
 \arccos|\langle\phi_t|\widehat{\psi}_t\rangle|\le gt.
 \label{sm:angle}
\end{equation}
Only the range $gt\le\pi/2$ will be needed.
One may equivalently extend all common isometries to unitaries on a fixed larger space and use the interaction picture there.

To turn the angle bound \eqref{sm:angle} into a testing bound, consider two unit vectors with overlap $\cos\theta$, where $0\le\theta\le\pi/2$.
Let $a$ and $b$ be the probabilities of accepting the null hypothesis on the first and second states, respectively, using the binary measurement $\{Q,\1-Q\}$, where outcome $Q$ means accepting the null hypothesis.
If $a\ge1-\eps$ and $\theta<\arccos\sqrt\eps$, then $b\ge\left[\sqrt{1-\eps}\cos\theta-\sqrt\eps\sin\theta\right]^2$.
To verify this directly, write the two acceptance probabilities as
$a=\cos^2 A$ and $b=\cos^2 B$, with $A,B\in[0,\pi/2]$.
Monotonicity of fidelity under the binary measurement gives
$\cos\theta\le\sqrt{ab}+\sqrt{(1-a)(1-b)}=\cos(A-B)$.
Hence $|A-B|\le\theta$.
Since $A\le\arcsin\sqrt\eps$, it follows that
$B\le\theta+\arcsin\sqrt\eps<\pi/2$, which gives the claimed lower bound on $b$.

Apply this bound to $\ket{\phi_t},\ket{\widehat{\psi}_t}$ and lift the original final effect to the purification space.
The actual output under the alternative hypothesis dominates the no-jump substate, so \eqref{sm:survival} and \eqref{sm:angle} imply, for $t<T_\eps$,
\begin{equation}
 \beta_{\eps,t}(0\|\cL)\ge
 e^{-\nu t}
 \left[\sqrt{1-\eps}\cos(gt)-\sqrt\eps\sin(gt)\right]^2>0.
 \label{sm:thresholdfloor}
\end{equation}
The same lower bound holds at all positive resolution times.
Together with endpoint achievability, this proves \eqref{sm:threshold}.
This completes the proof of Theorem~\ref{thm:identity}.
\end{proof}

For $g>0$, Theorem~\ref{thm:identity} also determines the threshold for the optimal symmetric error probability with equal priors:
\begin{equation}
 \perr(0,\cL)=0\quad\Longleftrightarrow\quad t\ge\frac{\pi}{2g}.
 \label{sm:symmetricthreshold}
\end{equation}
Indeed, at $t\ge\pi/(2g)$, the case $\eps=0$ of \eqref{sm:threshold} allows $\beta$ to approach zero with exactly $\alpha=0$, so $\perr(0,\cL)=0$.
If $t<\pi/(2g)$, choose a fixed $0<\eps<1$ small enough that $t<T_\eps$.
Any sequence with both errors tending to zero would eventually satisfy $\alpha\le\eps$, contradicting $\beta_{\eps,t}(0\|\cL)>0$ from \eqref{sm:threshold}.
This proves \eqref{sm:symmetricthreshold}.

When $g=0$, finite $\gamma(0\|\cL)$ gives $\beta_{\eps,t}(0\|\cL)>0$ for all $0\le\eps<1$ and $t\ge0$.
It also gives a uniform symmetric floor: if $\alpha\ge1/2$, the average error is at least $1/4$. Otherwise the entropy budget yields
$\beta\ge2^{-2(\gamma(0\|\cL) t+1)}$.
Thus $\perr(0,\cL)\ge2^{-2\gamma(0\|\cL) t}/8>0$.

For $g>0$ and fixed $t<T_\eps$, \eqref{sm:thresholdfloor} is independent of $\smin$.
In particular, $-\log\bmin(0\|\cL)=O_{\cL,t,\eps}(1)$ as $\smin\to0$.
The logarithmic law for the resolution time in Sec.~\ref{sec:log} holds for sufficiently large total time and must not be extended below this exact threshold.

When the alternative hypothesis is purely Hamiltonian, $G=-iH$ and $\cS=\mathbb C\1$.
The closest scalar to $-iH$ is the imaginary midpoint of its extreme eigenvalues, giving
$g=(\lambda_{\max}(H)-\lambda_{\min}(H))/2$.
At $\eps=0$, \eqref{sm:threshold} reduces to $T_0=\pi/(\lambda_{\max}(H)-\lambda_{\min}(H))$.
An equal superposition of the corresponding two eigenvectors attains this time by uninterrupted unitary evolution, so here a finite protocol can attain the zero~\cite{Childs2000}.

\section{Examples of the three regimes}\label{sec:qubit}

We now verify the examples in Section~\ref{sec:letter-examples}.
Taking as null hypothesis the identity semigroup and as alternative hypothesis the depolarising semigroup gives a finite Stein exponent.
Taking as null hypothesis the depolarising semigroup and as alternative hypothesis the identity semigroup instead allows zero type-II error at positive resolution time.
Amplitude damping also admits a finite support test.
With both excitation and relaxation present, the same type of drift imbalance instead gives an optimal type-II error probability that tends to zero only as the resolution time vanishes.

\subsection{Depolarising semigroup}\label{sec:depolarising}

We first take as null hypothesis the identity semigroup and as alternative hypothesis the semigroup with generator
$\cL_{\mathrm{dep}}(X)=\pi_A\Tr(X)-X$, where $\pi_A=\1_A/d$ is the maximally mixed state and $d\ge2$.
Since $(\cL_{\mathrm{dep}}\otimes\id_E)(\sigma_{AE})=\pi_A\otimes\sigma_E-\sigma_{AE}$, Theorem~\ref{thm:rate} gives
\begin{equation*}
 \gamma(0\|\cL_{\mathrm{dep}})=\sup_{E,\rho_{AE},\sigma_{AE}>0}\Tr\left[\rho_{AE}\,\Omega_{\sigma_{AE}}\left(\sigma_{AE}-\pi_A\otimes\sigma_E\right)\right]=\sup_{E,\sigma_{AE}>0}\left\{\log e-\lambda_{\min}\left(\Omega_{\sigma_{AE}}(\pi_A\otimes\sigma_E)\right)\right\},
\end{equation*}
where the second equality uses $\Omega_{\sigma_{AE}}(\sigma_{AE})=(\log e)\1_{AE}$ and the optimisation over $\rho_{AE}$.
To evaluate the remaining supremum, note that every bipartite state satisfies
\begin{equation*}
 \sigma_{AE}\le d\,\1_A\otimes\sigma_E
 =d^2\pi_A\otimes\sigma_E.
\end{equation*}
Applying the positive map $\Omega_{\sigma_{AE}}$ to this inequality gives
\begin{equation*}
 \Omega_{\sigma_{AE}}(\pi_A\otimes\sigma_E)
 \ge\frac{\log e}{d^2}\1_{AE},
 \qquad
 \lambda_{\min}\left(\Omega_{\sigma_{AE}}(\pi_A\otimes\sigma_E)\right)
 \ge\frac{\log e}{d^2}.
\end{equation*}
The bound is approached by taking $E\simeq A$ and
$\sigma_{AE}^{(\delta)}=(1-\delta)\kb{\Phi}+\delta\1_{AE}/d^2$, where $\ket{\Phi}$ is maximally entangled and $0<\delta<1$.
Indeed, $\sigma_E^{(\delta)}=\1_E/d$ and
\begin{equation*}
 \lambda_{\min}\left(\Omega_{\sigma_{AE}^{(\delta)}}
 (\pi_A\otimes\sigma_E^{(\delta)})\right)
 =\frac{\log e}{d^2(1-\delta+\delta/d^2)}
 \xrightarrow[\delta\to0]{}\frac{\log e}{d^2}.
\end{equation*}
Hence
\begin{equation*}
 \gamma(0\|\cL_{\mathrm{dep}})=(1-d^{-2})\log e.
\end{equation*}
The positive lower bound in \eqref{eq:finitefloor} therefore applies for every finite total time $t\ge0$ and $0\le\eps<1$, placing this example in Regime~1 of Table~\ref{tab:regimes}.

Reversely, we now take as null hypothesis the depolarising semigroup and as alternative hypothesis the identity semigroup.
The same Bell input gives a full-rank output under the null hypothesis and the pure state $\kb{\Phi}$ under the alternative hypothesis.
Their relative entropy is infinite for every $s>0$, so \eqref{sm:diff} gives $\gamma(\cL_{\mathrm{dep}}\|0)=\infty$.
The support violation places this pair of hypotheses in Regime~2 by Corollary~\ref{cor:trichotomy}.

\subsection{Amplitude damping}\label{sec:amplitude-damping}

Let the null hypothesis be the identity semigroup and let the generator under the alternative hypothesis be $\cLgad{\mu}{0}$, with $\mu>0$ and jump and drift operators
\begin{equation*}
 L_-=\sqrt\mu\ket0\!\bra1,\qquad G=-\frac{\mu}{2}\kb1.
\end{equation*}
This is the usual amplitude damping semigroup, corresponding to $\mu_+=0$ and $\mu_-=\mu$ in the generalised amplitude damping family considered in Section~\ref{sec:letter-gad}.
For $s>0$, the channel $e^{s\cLgad{\mu}{0}}$ has Kraus operators $K_0=\kb0+e^{-\mu s/2}\kb1$ and $K_1=\sqrt{1-e^{-\mu s}}\ket0\!\bra1$.
Since $e^{-\mu s/2}<1$, the diagonal entries of $K_0$ differ, while $K_1$ is off-diagonal, so $\1\notin\Span\{K_0,K_1\}$.
The Choi support is the span of the vectorised Kraus operators, and hence
\begin{equation*}
 \supp J_{AA'}(\id_A)\not\subseteq\supp J_{AA'}\left(e^{s\cLgad{\mu}{0}}\right).
\end{equation*}
Lemma~\ref{lem:support-attainment} therefore gives, for every $\smin>0$ and $0<\eps<1$, a protocol $\mathsf P$ with finite total exposure time and errors $\alpha(\mathsf P)\le\eps$, $\beta(\mathsf P)=0$.
This places the pair in Regime~2 of Table~\ref{tab:regimes}, with $\gamma(0\|\cLgad{\mu}{0})=\infty$ by Corollary~\ref{cor:trichotomy}.
Exchanging $\ket0$ and $\ket1$ gives the same conclusion when $\mu_-=0$ and $\mu_+>0$.

\subsection{Generalised amplitude damping}\label{sec:generalised-amplitude-damping}

Consider the generator $\cLgad{\mu_-}{\mu_+}$ with jump operators $L_-=\sqrt{\mu_-}\ket0\!\bra1$ and
$L_+=\sqrt{\mu_+}\ket1\!\bra0$, with relaxation rate $\mu_->0$ and excitation rate $\mu_+>0$, and no Hamiltonian part.
We first assume $\mu_+\ne\mu_-$ and return to equal transition rates below.
Then
\begin{equation*}
 G=-\tfrac12\left(\mu_+\kb0+\mu_-\kb1\right),\qquad
 \cS=\Span\{\1,\ket0\!\bra1,\ket1\!\bra0\}.
\end{equation*}
With the identity semigroup as the null hypothesis, we now evaluate $g$ in \eqref{sm:g} to determine the threshold $T_\eps$ at which the unrestricted optimal type-II error probability vanishes, using Theorem~\ref{thm:identity}.
For $C=a\1+b\ket0\!\bra1+c\ket1\!\bra0$, the norm of $G-C$ is at least the maximum modulus of its two diagonal entries, and hence at least $|\mu_+-\mu_-|/4$.
Taking $C=-(\mu_++\mu_-)\1/4$ attains this value.
If exactly one transition rate vanishes, $\cS$ loses one off-diagonal generator, but the same scalar $C$ still attains this lower bound.
Thus, for all unequal nonnegative transition rates, Theorem~\ref{thm:identity} gives
\begin{equation}
 g=\frac{|\mu_+-\mu_-|}{4},\qquad
 T_\eps=\frac{4\arccos\sqrt\eps}{|\mu_+-\mu_-|}.
 \label{sm:gadg}
\end{equation}
In particular, at total exposure time $T_0=2\pi/|\mu_+-\mu_-|$, the protocol of Section~\ref{sec:zeno-gad-motivation} has $\alpha=0$ and $\beta=O(\smin)$ by \eqref{sm:identity-block} and \eqref{sm:block-resolution}.

Next, we show that, for unequal positive transition rates, discriminating the identity semigroup from the generalised amplitude damping semigroup lies in Regime~3 of Table~\ref{tab:regimes}.
Since $G$ has unequal diagonal entries we have $G\notin\cS$, so Theorem~\ref{thm:span} gives $\gamma(0\|\cLgad{\mu_-}{\mu_+})=\infty$.
We next compare the supports of the Choi states of the dynamics under the two hypotheses and put $\eta=e^{-(\mu_++\mu_-)s}$ and $p=\mu_-/(\mu_++\mu_-)$.
A Kraus representation of the channel $e^{s\cLgad{\mu_-}{\mu_+}}$ is
\begin{equation*}
 \sqrt p\left(\kb0+\sqrt\eta\kb1\right),\quad
 \sqrt{p(1-\eta)}\ket0\!\bra1,\quad
 \sqrt{1-p}\left(\sqrt\eta\kb0+\kb1\right),\quad
 \sqrt{(1-p)(1-\eta)}\ket1\!\bra0 .
\end{equation*}
For $s>0$ and positive transition rates, all four Kraus operators are linearly independent.
The two off-diagonal operators are independent, and the two diagonal vectors have nonzero determinant proportional to $1-\eta$.
The Choi support is the span of the vectorised Kraus operators, so
\begin{equation*}
 \supp J_{AA'}(\id_A)\subseteq\supp J_{AA'}\left(e^{s\cLgad{\mu_-}{\mu_+}}\right)=\mathcal{H}_A\otimes\mathcal{H}_{A'}
 \qquad(s>0).
\end{equation*}
Together with the infinite Stein exponent, this establishes Regime~3 of Table~\ref{tab:regimes} by Corollary~\ref{cor:trichotomy}.
For every fixed $0<\eps<1$, Theorem~\ref{thm:resolution} gives $\bmin(0\|\cLgad{\mu_-}{\mu_+})>0$ at every finite total exposure time and positive resolution time.
By Theorem~\ref{thm:identity}, this optimal error nevertheless tends to zero as $\smin\to0$ whenever $t\ge T_\eps$, with $T_\eps$ given in \eqref{sm:gadg}.

\subsubsection{Equal transition rates}\label{sec:equal-transition-rates}

To complete the classification, let $\mu_+=\mu_-=\mu\ge0$, including the origin.
Then $G=-\mu\1/2\in\cS$, so Theorem~\ref{thm:span} gives $\gamma(0\|\cLgad{\mu}{\mu})<\infty$.
The bound \eqref{eq:finitefloor} therefore gives a positive optimal type-II error probability at every finite total exposure time for $0\le\eps<1$, placing this case in Regime~1 of Table~\ref{tab:regimes}.

\section{Classical comparison}\label{sec:classical}

In the quantum case, Theorem~\ref{thm:identity} shows that, when the null hypothesis is the identity semigroup and $g>0$, both error probabilities can approach zero at a common finite total exposure time.
We now show why this conclusion has no classical counterpart with bounded exit rates, and why it is stronger than the general asymmetric classification.

On a finite state space, a continuous-time Markov generator is a real matrix $L=(L_{yx})$ with $L_{yx}\ge0$ for $y\ne x$ and $\sum_yL_{yx}=0$ for every $x$.
The generator $L$ defines a continuous semigroup $(e^{tL})_{t\ge0}$ of stochastic matrices, whose entries are nonnegative and whose columns sum to one~\cite[Section~2.1]{Norris1997}.
Probability column vectors evolve as $p(t)=e^{tL}p(0)$.
For $y\ne x$, the entry $L_{yx}$ is the transition rate from $x$ to $y$.
The exit rate from state $x$ is $r(x)=\sum_{y\ne x}L_{yx}=-L_{xx}$.

A classical protocol prepares an initial system state, possibly at random, and may use a finite classical memory to store observations.
It alternates evolution under the unknown Markov generator with stochastic controls, which can randomly update the system and memory or reset the system.
The controls and evolution durations may depend on previous observations, but the total exposure time is at most $t$ on every run.
At the end, a possibly randomised decision selects the null hypothesis or the alternative hypothesis.
The same preparation, control and decision rules are used under both hypotheses.

\begin{proposition}\label{prop:classical}
For two finite-state continuous-time Markov generators $L_1,L_2$ with exit rates $r_i(x)$, let
$\Lambda=\max_{i,x}r_i(x)$.
Every classical protocol of pathwise total time at most $t$ satisfies
\begin{equation}
 \alpha+\beta\ge e^{-\Lambda t},\qquad
 p^{\mathrm{class}}_{\mathrm{err},t}(L_1,L_2)\ge\tfrac12e^{-\Lambda t}.
 \label{sm:classical}
\end{equation}
The statement allows arbitrary finite classical memory, stochastic controls, resets, adaptive observation times and even observation of the entire physical jump record.
\end{proposition}

\begin{proof}
For a fixed protocol, let $P_i(h)$ be the probability of accepting hypothesis $h\in\{1,2\}$ under generator $L_i$, and let $P_0(h)$ be the same probability with the unknown dynamics replaced by the identity.
An identity run has evolution durations $s_k$ and states $x_k$ during these intervals.
Under hypothesis $i$, the probability of no physical jump along these intervals is~\cite[Section~2.6]{Norris1997}
\begin{equation*}
 \prod_k e^{-r_i(x_k)s_k}
 =e^{-\sum_k r_i(x_k)s_k}
 \ge e^{-\Lambda\sum_k s_k}\ge e^{-\Lambda t}.
\end{equation*}
Weighting each run under identity dynamics by its no-jump probability and averaging over the random preparations, controls and final decision gives $P_i(h)\ge e^{-\Lambda t}P_0(h)$.
This remains valid for adaptive controls and durations, since $\sum_k s_k\le t$ on every run.
Consequently,
\begin{equation*}
 \alpha+\beta=P_1(2)+P_2(1)
 \ge e^{-\Lambda t}\left(P_0(2)+P_0(1)\right)=e^{-\Lambda t}.
\end{equation*}
Taking the infimum of $(\alpha+\beta)/2$ over protocols proves the symmetric error bound.
\end{proof}

We now determine the unrestricted optimal type-II error probability when the null hypothesis is the identity semigroup, $L_1=0$, and the generator under the alternative hypothesis is $L_2=L$.
Fix $0\le\eps<1$ and total exposure time $t$, and let $\Lambda=\max_x r(x)$ be the largest exit rate of $L$.
The preceding proof then gives, for every protocol with $\alpha\le\eps$,
\begin{equation*}
 \beta=P_2(1)\ge e^{-\Lambda t}P_1(1)
 =(1-\alpha)e^{-\Lambda t}\ge(1-\eps)e^{-\Lambda t}.
\end{equation*}

To show that this lower bound is optimal, we construct protocols whose type-II errors approach it as the observation interval tends to zero at fixed $t$.
Prepare a state $x$ with exit rate $r(x)=\Lambda$ and observe the state after each of $n$ evolution intervals of duration $t/n$.
Accept the alternative hypothesis if any observation differs from $x$.
If all observations return $x$, accept the null hypothesis with probability $1-\eps$ and the alternative hypothesis otherwise.
Under the null hypothesis, every observation returns $x$, so $\alpha_n=\eps$.
Under the alternative hypothesis, the probability of observing $x$ given that the previous state was $x$ is $(e^{tL/n})_{xx}$.
The Markov property and $L_{xx}=-\Lambda$ therefore give
\begin{equation*}
 \beta_n=(1-\eps)\left[(e^{tL/n})_{xx}\right]^n
 =(1-\eps)\left[1-\frac{\Lambda t}{n}+O(n^{-2})\right]^n
 \xrightarrow[n\to\infty]{}(1-\eps)e^{-\Lambda t}.
\end{equation*}
Each protocol uses total exposure time $t$ and finitely many observations, with $t/n\to0$ as $n\to\infty$.
Combining this construction with the lower bound proves
\begin{equation*}
 \beta^{\mathrm{class}}_{\eps,t}(0\|L)=(1-\eps)e^{-\Lambda t}.
\end{equation*}
For a process on two states with excitation rate $\mu_+$ and relaxation rate $\mu_-$, the value is $\Lambda=\max\{\mu_+,\mu_-\}$.
Consequently, when the null hypothesis is the identity semigroup, the classical Stein exponent is finite even when these transition rates differ.

\end{document}